%% file: Optimizations-for-Differential-Dataflow.tex
\documentclass[sigconf, nonacm]{acmart}

\newcommand\vldbdoi{10.14778/3551793.3551862}
\newcommand\vldbpages{3186 - 3198}
\newcommand\vldbvolume{15}
\newcommand\vldbissue{11}
\newcommand\vldbyear{2022}
\newcommand\vldbauthors{\authors}
\newcommand\vldbtitle{\shorttitle} 
\newcommand\vldbavailabilityurl{https://github.com/khaledammar/optimized-DC}
\newcommand\vldbpagestyle{empty} 

\newif\iftechreport
\techreporttrue

\input{common-includes}

\begin{document}
\title{Optimizing Differentially-Maintained Recursive Queries on Dynamic Graphs}

\author{Khaled Ammar}
\affiliation{%
  \institution{University of Waterloo,~BorealisAI}
  \city{Waterloo, ON}
  \country{Canada}
}
\email{khaled.ammar@uwaterloo.ca}

\author{Siddhartha Sahu}
\affiliation{
  \institution{University of Waterloo}
  \city{Waterloo, ON}
  \country{Canada}
}
\email{s3sahu@uwaterloo.ca}

\author{Semih Salihoglu}
\affiliation{
  \institution{University of Waterloo}
  \city{Waterloo, ON}
  \country{Canada}
}
\email{semih.salihoglu@uwaterloo.ca}

\author{M. Tamer \"{O}zsu}
\affiliation{
  \institution{University of Waterloo}
  \city{Waterloo, ON}
  \country{Canada}
}
\email{tamer.ozsu@uwaterloo.ca}

\begin{abstract}
Differential computation (DC) is a highly  general incremental computation/view maintenance technique that can maintain the output of an arbitrary and possibly recursive dataflow computation upon changes to its base inputs. As such, it is a promising technique for graph database management systems (GDBMS) that  support continuous 
recursive queries over dynamic graphs.  Although differential computation can be highly efficient for maintaining these queries, it can require prohibitively large amount of memory. This paper studies how to reduce the memory overhead of DC with the goal of increasing the  scalability of systems that adopt it. We propose a suite of optimizations that are based on dropping the differences of operators, both completely or partially, and recomputing these differences when necessary. We propose deterministic and probabilistic data structures to keep track of the dropped differences. Extensive experiments demonstrate that the optimizations can improve the scalability of a DC-based continuous query processor.
\end{abstract}

\maketitle

\pagestyle{\vldbpagestyle}
\begingroup\small\noindent\raggedright\textbf{A shorter version of this paper appears in:}\\
\vldbauthors. \vldbtitle. PVLDB, \vldbvolume(\vldbissue): \vldbpages, \vldbyear.\\
\href{https://doi.org/\vldbdoi}{doi:\vldbdoi}
\endgroup


\input{introduction}

\input{related-work}

\input{preliminaries}
\input{complete-dropping}
\input{partial-dropping}

\input{evaluation}

\input{future-work}

\balance
\bibliographystyle{ACM-Reference-Format}
\bibliography{publications,csp,uw-ethesis}
\iftechreport
\appendix
\input{Appendix}

\input{implementation-short}
\fi

\end{document}

%% file: common-includes.tex
\usepackage[T1]{fontenc}
\usepackage[utf8]{inputenc}
\usepackage{booktabs}
\usepackage{multirow}
\usepackage{balance}
\usepackage{xcolor} 

\usepackage{verbatim} 
\usepackage{graphicx}
\usepackage{caption}
\usepackage{subcaption} 
\usepackage{setspace} 
\usepackage{float}
\usepackage{xcolor,colortbl}
\usepackage[normalem]{ulem}
\usepackage{tikz}
\usepackage{pgfplots,pgfplotstable} 
\usetikzlibrary{patterns}
\usetikzlibrary{matrix,decorations.pathreplacing, calc, positioning}
\usetikzlibrary{graphs,graphs.standard,quotes}
\usetikzlibrary{arrows, arrows.meta, automata}
\usepgfplotslibrary{groupplots}

\pgfplotsset{compat=1.16,
  /pgfplots/ybar legend/.style={
  /pgfplots/legend image code/.code={%
      \draw[##1,/tikz/.cd,yshift=-0.2em]
      (0cm,0cm) rectangle (10pt,4pt);},
  },
}
\definecolor{ctorange}{RGB}{213, 94, 0}
\definecolor{ctblue}{RGB}{0, 114, 178}
\definecolor{ctyellow}{RGB}{240, 228, 66}
\definecolor{ctgreen}{RGB}{0, 158, 115}
\definecolor{ctdarkyellow}{RGB}{230, 159, 0}
\pgfplotscreateplotcyclelist{three-colors}{
    {draw=none,fill=ctorange},
    {draw=none,fill=ctyellow},
    {draw=none,fill=ctblue},
}
\pgfplotscreateplotcyclelist{noadapt-adapt}{
    {draw=none,fill=ctorange},
    {draw=none,fill=ctyellow},
}
\pgfplotscreateplotcyclelist{gs-gb}{
    {draw=none,fill=ctdarkyellow},
    {draw=none,fill=ctgreen},
}

\usepackage[ruled,vlined]{algorithm2e}
\usepackage{listings}
\definecolor{key-color}{rgb}{0.8, 0.47, 0.196}
\usepackage{enumitem} 
\usepackage{hyphenat} 

\usepackage{doi} 
\usepackage{hyperref}
\expandafter\def\expandafter\UrlBreaks\expandafter{\UrlBreaks\do\/\do\-\do\.} 
\definecolor[named]{Purple}{cmyk}{0.55,1,0,0.15}
\definecolor[named]{DarkBlue}{cmyk}{1,0.58,0,0.21}
\hypersetup{bookmarksnumbered,unicode,naturalnames,colorlinks,breaklinks,
linkcolor=Purple,
citecolor=Purple,
urlcolor=DarkBlue,
filecolor=DarkBlue}

\usepackage{array}
\newcolumntype{L}[1]{>{\raggedright\let\newline\\\arraybackslash\hspace{0pt}}m{#1}}

\newenvironment{squishedlist}
{
  \begin{list}{$\bullet$}
   {
     \setlength{\itemsep}{0pt}
     \setlength{\parsep}{2pt}
     \setlength{\topsep}{1.0pt}
     \setlength{\partopsep}{0pt}
     \setlength{\leftmargin}{01.5em}
     \setlength{\labelwidth}{1em}
     \setlength{\labelsep}{0.5em}
   }
}
{
   \end{list}
}
\newenvironment{squishedenumerate}
{
  \begin{enumerate}[label=(\roman*),topsep=0.2em,itemsep=0.2em,leftmargin=1.7em] 
}
{
  \end{enumerate}
}

\newcommand{\noun}[1]{\textsc{#1}}
\newcommand{\graphsurge}[1]{\noun{Graphsurge}}

\newcommand{\diffonly}[1]{\texttt{diff-only}}
\newcommand{\scratch}[1]{\texttt{scratch}}

\newlength{\subsubheadingvspace}
\newlength{\imagebottommargin}
\newtheorem{theorem}{Theorem}[section]

\newcounter{theorem2}
\newtheorem{example}[theorem2]{Example}

\newcommand{\JOIN}{\texttt{Join}}

\newcommand{\DC}{\textsc{DC}}
\newcommand{\VDC}{\textsc{VDC}}
\newcommand{\IFE}{\textsc{IFE}}
\newcommand{\JOD}{\textsc{JOD}}

\newcommand{\DET}{\textsc{Det-Drop}}
\newcommand{\PROB}{\textsc{Prob-Drop}}
\newcommand{\DD}{\textsc{DD}}
\newcommand{\SC}{\textsc{Scratch}}
\newcommand{\diffs}{differences}
\newcommand{\diff}{difference}
\newcommand{\Random}{\texttt{Random}}
\newcommand{\Degree}{\texttt{Degree}}
\newcommand{\Tau}{\mathrm{T}}

%% file: introduction.tex



\section{Introduction}
\label{sec:introduction}

Graph queries that are recursive in nature, such as single pair shortest path (SPSP), single source shortest path (SSSP),
variable-length join queries, or regular path queries (RPQ), are prevalent across applications that are developed on graph database management systems (GDBMS). Many of these applications require maintaining query results incrementally, as the graphs stored in GDBMSs are dynamic and evolve over time. For example, millions of travellers use navigation systems to find the fastest route between two points on a map. To keep the route information fresh, these systems need to continuously update their SPSP query results as road conditions change. Similarly, several knowledge graphs, such as  RefinitivGraph~\cite{refinitivGraph} contain billions of connections between real-world entities, such as companies, banks, stocks, and managers. A Refinitiv product, World-Check Risk Intelligence\footnote{https://www.refinitiv.com/en/products/world-check-kyc-screening}, searches for direct and indirect connections between entities to help companies and banks comply with mandatory regulations. Since these graphs are frequently updated by new facts, these applications require the queries to be continuously evaluated.

Many GDBMSs have capabilities to evaluate one-time versions of recursive queries over static graphs, but generally do not support incrementally maintaining them. As such, in dynamic graphs, existing systems require rerunning these queries from scratch at the application layer. A GDBMS that can incrementally maintain recursive queries inside the system would lead to easier and more efficient application development. In this paper we investigate the use of \emph{differential computation} (DC)~\cite{mcsherry:ddf}, a new incremental maintenance technique, to maintain the results of recursive queries in GDBMSs. DC is designed to maintain arbitrarily cyclic (thus, recursive) dataflow programs~\cite{mcsherry:ddf, naiad}. 

Unlike using a specialized incremental derivation rule, DC starts from a dataflow program that evaluates the one-time version of the query. By keeping track of the differences to the inputs and outputs of the operators across different iterations,
called {\em timestamps} in DC terminology,
 DC maintains and propagates the changes between operators as the original inputs to the dataflow are updated. This makes DC more general than other techniques, as it is agnostic to the underlying dataflow computation. 

However, DC can have significant memory overhead~\cite{iturbograph}, as it may need to monitor a high number of input and output differences between operators. For example, Table~\ref{table:DiffScalabilityIssue} shows the performance and memory overhead of the DC implementation of the standard Bellman-Ford algorithm for maintaining the results of SSSP queries on the Skitter internet topology dataset~\cite{snapDatasets}. In the experiment, we modify the graph with 100 batches of 1 random edge insertion each, and provide the system with 10GB memory to store the generated differences. The table also shows the performance of  a baseline that re-executes the Bellman-Ford algorithm from scratch after each update, thus not requiring any memory for maintaining these queries. Although the differential version of the algorithm is about five orders of magnitude faster, it cannot maintain more than 10 concurrent queries due to its large memory requirement. This limits the scalability of systems that adopt DC.

\begin{table}[t]
    \centering
    \iftechreport
    \caption{Execution time (in seconds) for an SPSP workload, on Skitter dataset, using a scratch algorithm, which re-executes a standard non-incremental Bellman-Ford algorithm, vs.
a differential computation version, which keeps track of changes. Differential computation is more than five order of magnitude faster, but fail with out of memory when the number of queries increases.}
    \fi
    \vspace{-10pt}
    \begin{tabular}{|l||c|c|c|c|}
        \hline
        Number of Queries & 10 & 20 & 30 & 40 \\
        \hline
        \hline
        Scratch & $6.1K$ &$13.6K$ & $20.7K$ & $28.3K$ \\
        \hline
        Differential Computation  & $0.2$ & OOM & OOM & OOM \\
        \hline
    \end{tabular}
    \label{table:DiffScalabilityIssue}
        \vspace{-10pt}
\end{table}

In this paper, we study how to reduce the memory overheads of DC to increase its scalability when maintaining the popular classes of recursive queries mentioned above. Our optimizations are broadly based on \emph{dropping differences} and instead recomputing them when necessary. We focus on optimizing the differential version of a common subroutine in graph algorithms where vertices aggregate their neighbours' values iteratively and propagate their own values to neighbours until a stopping condition, such as a fixed point, is reached. Variants of this subroutine with different aggregation, propagation, and stopping conditions can be used to evaluate 
all of the recursive queries we focus on in this paper. This routine consists of \JOIN\ operator and an aggregation operator, e..g, a \texttt{Min}, and has been given different terms in literature, such as \emph{propagateAndAggregate}~\cite{salihoglu:help} or \emph{iterative matrix vector multiplication}~\cite{kang:pegasus}. We refer to it as \emph{iterative frontier expansion} (\IFE).

In this work, we start with the base version implementation of DC as in the  differential dataflow (\DD)~\cite{mcsherry:ddf} and its precursor Naiad system~\cite{naiad}. 
We propose two main optimizations: \textsc{Join-On-Demand} (\JOD) (Section~\ref{sec:complete-drop}) that
completely drops the output differences of the \JOIN\ operator of the \IFE\ dataflow and only computes these differences when DC needs to inspect them; and (2) two {\em partial difference dropping} optimizations (Section ~\ref{sec:partial}) that
 drop some of the differences in the output of the aggregation operator in \IFE. 

Our partial difference dropping optimization offers users a knob to drop a certain percentage of the system's \diffs.
We begin by describing a baseline deterministic optimization \textsc{Det-Drop} that explicitly keeps track of the vertex and timestamp of each dropped difference. We show that although \textsc{Det-Drop} reduces the memory consumption of a system, it also has inherent limitations in terms of scalability improvements, as the additional state it keeps is proportional to the amount of differences that it drops. We then propose a probabilistic approach \textsc{Prob-Drop} that addresses this shortcoming by leveraging a probabilistic data structure, specifically a Bloom filter.
\textsc{Prob-Drop} may attempt to reconstruct a non-existing difference due to false negatives but it more effectively reduces the memory consumption, so a system using \textsc{Prob-Drop} needs to drop fewer differences to meet same memory budgets as \textsc{Det-Drop}.
Finally, we describe an optimization that uses the degree information of each vertex to choose which \diffs\ to drop as opposed to dropping them randomly.

We demonstrate
that \JOD\ reduces the number of differences up to $8.2\times$ in comparison to vanilla \DC\ implementations. We also show that
exploiting the degree information to select  the differences to drop can improve the performance of partial
dropping optimizations (\textsc{Det-Drop} or \textsc{Prob-Drop}) by several orders of magnitude.
We further show that \textsc{Prob-Drop} achieves up to $1.5\times$ scalability relative to \textsc{Det-Drop} when selecting the differences to drop based on degrees. Our optimizations overall can increase the scalability of our differential algorithms by up to $20\times$ in comparison to \DD, while still outperforming a baseline that reruns computations from scratch by several orders of magnitude.

%% file: related-work.tex
\section{Related Work}
\label{sec:rw}

Broadly, there are two approaches to maintaining the results of a computation over a dynamic graph: (i) using a computation-specific specialized solution; or (ii) using a generic incremental computation/view maintenance solution that is oblivious to the actual computation, at least for some class of computations. DC falls under the second category. Below, we review both approaches.

\subsection{Specialized Techniques and Systems}
There is extensive literature dating back to 1960s on developing specialized incremental versions of (aka {\em dynamic}) graph algorithms that maintain their outputs as an input graph changes. Many of the earlier work focuses on versions of shortest path algorithms, in particular all pairs shortest paths computation~\cite{demetrescu2001fully, roditty2011dynamic, DynamicAPSP, APSP-Chen2005, APSP-Chen2012, networkEvaluation1967, rodionov1968parametric}. These works aim at developing fast algorithms that can, in worst-case time, be faster than recomputing shortest paths upon a single update, e.g., when the edge weights are integer values. 
\iftechreport
Fan et al~\cite{fan2017incremental} present theoretical results on the foundations of such algorithms. Specifically they show that the cost of performing six specific incremental graph computations, such as regular path queries and strongly connected components algorithms, cannot be bounded by only the size of the changes in the input and output. Then, they develop algorithms that have bounded guarantees in terms of the work performed to maintain the computation.
\fi

On the systems side, there are several graph analytics systems that enable users to develop incremental versions of a graph algorithm. GraphBolt~\cite{graphBolt} is a recent shared-memory parallel streaming system that can maintain dynamic versions of graph algorithms. 
\iftechreport
GraphBolt requires users to write explicit maintenance code in functions such as \texttt{retract} or \texttt{propagateDelta} that generic systems such as \DD\ do not require. As graph updates arrive, the system executes these functions, and if a user has provided a dynamic algorithm with provable convergence guarantees, the system will correctly maintain the results.
\fi

\iftechreport
iTurboGraph~\cite{iturbograph} focuses on incremental neighbour-centric graph analytics with an objective to reduce the overhead of large in-memory intermediate results in systems like GraphBolt and \DD. iTurboGraph keeps graph data on disk as streams and model the graph traversal as enumeration of walks to avoid maintaining large intermediate results in memory. They avoid expensive random disk access by adopting the nested graph windows approach~\cite{turboGraph}. Instead, our proposed solutions keep the intermediate results in memory and drop these differences to reduce memory overhead. 
\fi

Broadly, programming specialized algorithms or GraphBolt-like 
systems can be more efficient than generic solutions. For example, several references have demonstrated this difference between \DD\ and GraphBolt~\cite{graphBolt,  Graphsurge}. In contrast, generic solutions such as \DD, which we focus on in this work, are fundamentally different and have the advantage that users can program arbitrary static versions of their algorithms, which will be automatically maintained. 
Therefore they are suitable as core incremental view maintenance techniques to integrate in general data management systems.

\subsection{Generic Techniques and Systems}
\label{sec:IVM}

When an input graph is modeled as a set of relations and a graph algorithm is modeled as a query over these relations,  maintaining graph computation can be modeled as {\em incremental view maintenance}, where the view is the final output of the query. 
Traditional incremental view maintenance (IVM) techniques for recursive SQL and Datalog queries have focused on variants of incremental maintenance approaches~\cite{green2013datalog} such as Delete-and-Rederive, which consists of a set of delta-rules that can produce the changes in the outputs of queries upon changes to the base relations. 
\iftechreport
These rules can be highly inefficient as they first delete all derivations of updated/removed facts and then redrive them again using the updated facts, only to finally detect whether any of the deletions and/or additions have effects on the final result. 
\fi
This contrasts with DC as it does not store intermediate computations to speed up processing. Interestingly, the only incremental open-source Datalog implementation we are aware of does not use the Delete-Rederive maintenance algorithm but uses DC~\cite{ddlog}. This work compiles Datalog programs into \DD\ programs, so ultimately uses vanilla \DD, which we optimize and use as a baseline in our work.

Tegra~\cite{DBLP:conf/nsdi/IyerPPGS21} is a system developed on top of Apache Spark~\cite{zaharia:spark}, that is designed to perform ad-hoc window-based analytics on a dynamic graph. Tegra allows the creation of arbitrary snapshots of  graphs and executes computations on these snapshots. The system has a technique for sharing arbitrary computation across snapshots through a computation maintenance logic similar to DC. However, the system is optimized for retrieving arbitrary snapshots quickly instead of sharing computation across snapshots efficiently.

There has been several systems work that use the generic incremental maintenance capabilities of DC. GraphSurge~\cite{Graphsurge} is a distributed graph analytics system that lets users create multiple arbitrary views of a graph organized into a \textit{view collection} using a declarative \textit{view definition language}. Users can then run arbitrary computations on these views using a general programming API that uses \DD\ as its execution engine, which allows Graphsurge to automatically share computation when running across multiple views. References~\cite{stuecklberger:expressing} implements a DC-based Software Defined Network Controller that incrementally updates the routing logic as the underlying physical layer changes. Similarly, RealConfig~\cite{zhang:incremental} is a \textit{network configuration verifier} uses \DD\ to incrementally verify updates to a network configuration without having to restart from scratch after every change.

%% file: preliminaries.tex
\vspace{-7pt}
\section{Preliminaries}
\label{sec:preliminaries}
In this section we first review the graph and query models used in the paper. Then we summarize the IFE recursive algorithmic subroutine and differential computation. Table~\ref{tbl:queries} shows the notations and abbreviations 
that are used throughout the paper.

\begin{table}[t]
	\centering
	\caption{Commonly used notations and acronyms.}
	\vspace{-10pt}
	\begin{tabular}{|p{1.5cm} || p{6cm} |}
		\hline
		{\bf Acronym} & {\bf Description} \\
		\hline
		\hline
		$D$ and $\delta D$ & Vertex distance collections and their delta.\\
		\hline
		$E$ and $\delta E$ & Edges collections and their delta.\\
		\hline
		$J$ and $\delta J$ & The output of \texttt{Join} operator and their delta.\\
		\hline
		$C_{\langle G_k, i \rangle}$ & Collection $C$ at timestamp $\langle$ graph version $G_k$ and iteration $i$ $\rangle$.\\
		\hline
		$C^v$ & $C$'s partition by a key/vertex ID $v$.\\
		\hline
		\IFE\ & Iterative Frontier Expansion.\\
		\hline
		\DC\ & Differential Computation. \\
		\hline
		\VDC\ & Our Vanilla implementation for \DC.\ \\
		\hline
		\DD\ & Differential Dataflow system.\\
		\hline
		\JOD\ & Join-on-Demand optimization\\
		\hline
		\DET\ & Partial difference dropping optimization using a deterministic data structure.\\
		\hline
		\PROB\ & Partial difference dropping optimization using a probabilistic data structure.\\
		\hline
		\hline
	\end{tabular}
	\label{tbl:queries}
	\vspace{-15pt}
\end{table}

\vspace{-7pt}
\subsection{Graph and Query Model}
\label{sec:model}

We consider \textit{property graphs}, so vertices and edges can have attributes. Formally, a graph $G = (V, E, P_V, P_E)$, where $V$ is the set of vertices, $E$ is the set of directed edges, $P_V$ is the set of properties over vertices, and $P_E$ is the set of properties over edges. Our continuous queries  compute properties of vertices, which we refer to as their {\em states}. We will
not explicitly model states but these can be thought of as temporary properties in $P_V$.
 For an edge $e$, we maintain two properties:  $label(e)$, and $weight(e)$. If $G$ is unweighted, the the weights of each edge is set to $1$.

We focus on three recursive queries in this paper: SPSP, K-hop, and RPQ. K-hop is the query in which
we are given a source vertex $s$ and output all reachable vertices from $s$ that are at a distance (in terms of hops) of $\le k$ for a given
$k$. Each one of these queries can interact with different parts of our graph model. 
The edge properties that a recursive query needs to access and the vertex states for this computation will be clear from context. 

In a dynamic graph setting, an initial input graph $G_0$ may receive several batches of updates. Each batch is defined as a list of edge insertions or deletions $\delta E=[(u,v, label, weight, +/-)]$, which includes an edge, 
and its $label$~/~$weight$,
and a +/- to indicate, respectively, an insertion or a deletion (updates appear as one deletion and one insertion). 
We do not consider vertex insertions or deletions because these implicitly occur in our algorithms through explicit edge insertions and deletions.
$G_k$ refers 
to the actual set of edges in a graph $G$ after $G$ receives its $k$'th batch of updates $\delta E_k$ (so the union of $G_0$ and the $k$ batches of updates). 

The problem of incremental maintenance of a recursive query $Q$ is to 
report the changes to the output vertex
states of $Q$ after every batch of updates. These batches can be thought of as output in the form of $(v, state(v), +/-)$,
for a vertex $v$ and a new vertex state $state(v)$ and +/- indicating addition or removal of a state. 

\vspace{-10pt}
\subsection{Iterative Frontier Expansion as a Dataflow}
\label{sec:coreIFE}
Iterative Frontier Expansion (IFE) is a standard subroutine for implementing many graph algorithms solving many computational problems, including graph traversal queries like SPSP, SSSP, RPQs.  
At a high-level, the computation takes as input the edges (possibly with properties) of a graph $G$ and an initial set of vertex states, and, iteratively, aggregates for each vertex the states of its neighbours to compute a new vertex state, and propagates this state to its neighbours.
These iterations continue until some stopping criterion is met, e.g., a fixed point is reached and the vertex states converge. Figure~\ref{fig:ife-template} shows the template IFE dataflow that consists of two operators, \texttt{ExpandFrontier}, that expands the frontiers and the \texttt{Stop} operator that determines when to stop the query execution.

We use and optimize variants of this basic IFE dataflow to evaluate the queries we consider. As an example,
Figure~\ref{fig:sssp-as-ife} shows a specific instance of the IFE dataflow implementing the standard Belman-Ford algorithm for evaluating an SSSP query where vertex states are latest distances from a source vertex $s$. 
\texttt{ExpandFrontier} operator is implemented with two operators, \texttt{Join} and \texttt{Min}.
For each vertex $v$ in the frontier, \texttt{Join} sends possible new distances to $v$'s outgoing neighbours (considering $v$'s latest distance and possible weights on the edges). For each vertex $u$ of $v$'s outgoing neighbours', the new value is computed with a \texttt{Min} operator that computes the smallest received distance for $u$ considering $u$'s latest known distance. 
For different variants of shortest-path queries, RPQs, and variable-length join queries, we use the IFE template dataflow 
with always the same \texttt{Join} operator, but possibly different aggregator implementations and different \texttt{Stop} conditions. 

\begin{figure}
	\centering
	\begin{subfigure}{0.9\linewidth}
		\caption{High-level IFE dataflow.}
		\includegraphics[height=0.65in]{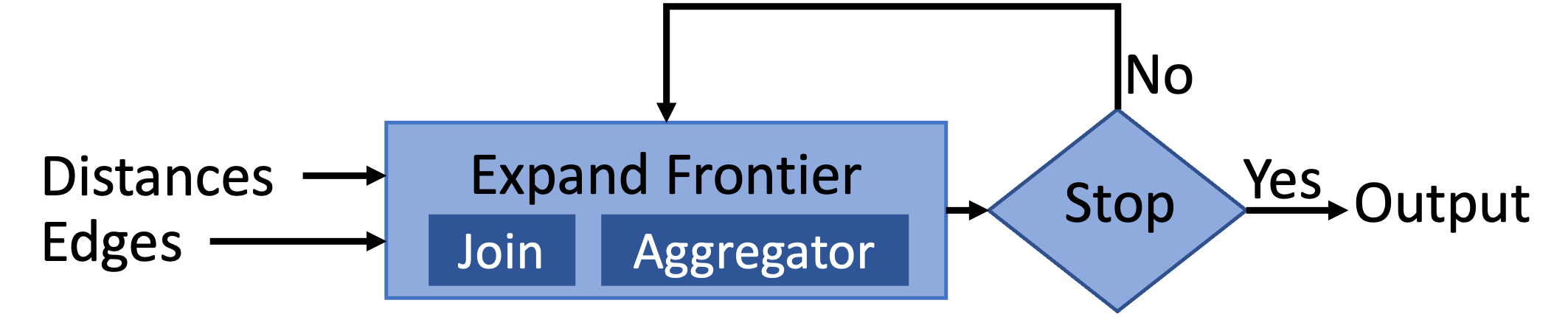}
		\label{fig:ife-template}
	\end{subfigure}%
	\vspace{-10pt}
	\newline
	\begin{subfigure} {0.9\linewidth}
		\caption{Specific IFE dataflow for SSSP}
		\includegraphics[height=0.7in]{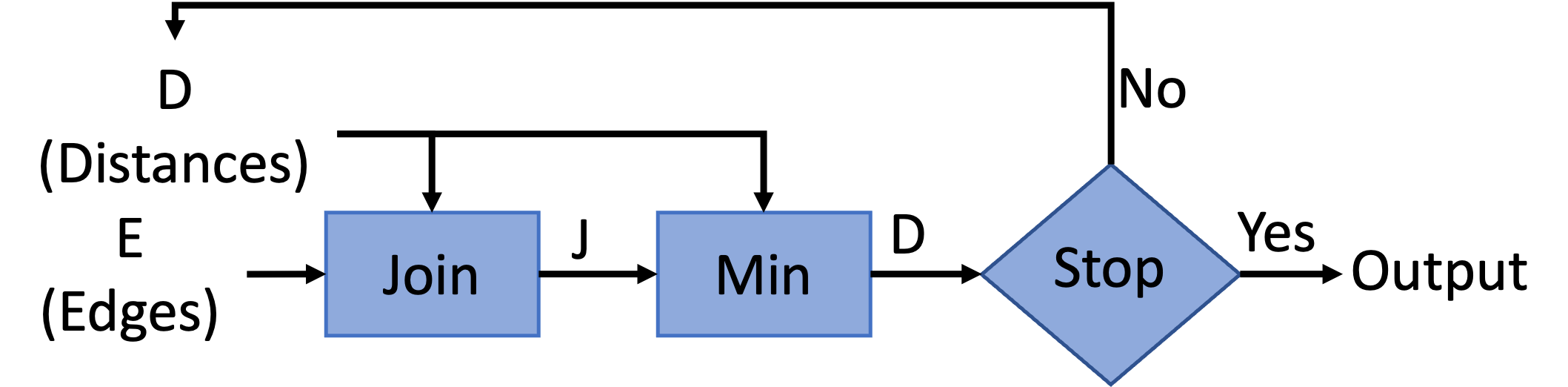}
		\label{fig:sssp-as-ife}
	\end{subfigure}
	\vspace{-10pt}

	\caption{Template IFE dataflow and a specific Bellman-Ford algorithm's dataflow implementation.}
	\label{fig:dataflow}
	\vspace{-10pt}
\end{figure}

\vspace{-7pt}
\subsection{Differential Computation Overview}
\label{sec:diff-overview}

DC~\cite{mcsherry:ddf} is a general technique to maintain the outputs of arbitrarily nested dataflow programs as the base input collections change.
Dataflow programs consist of operators, such
as \texttt{Join} or \texttt{Min} in Figure~\ref{fig:sssp-as-ife},
that take input and produce output
{\em data collections}, which are 
tables storing tuples. 
For example, in
the IFE dataflow,
the edges in an input graph are stored as (\texttt{src}, \texttt{dst}) tuples in the \texttt{Edges} ($E$) data collection.
We will refer to collections, such as \texttt{E},
that are inputs to the dataflow as {\em base collections}, and other collections
that are outputs of an operator  as {\em intermediate collections}.

We review DC through an example. Consider the IFE instance from Figure~\ref{fig:sssp-as-ife} implementing the Bellman-Ford algorithm and running it on the input graph shown in Figure~\ref{fig:running-ex-input}. 
Given this iterative dataflow computation, DC computes the input and output data collections of each operator as {\em partially ordered timestamped difference sets} and maintains these difference sets as the original input collections to the entire dataflow (in this case \texttt{Edges} ($E$) and \texttt{Distances} ($D$)) change. 
Timestamps can be multi-dimensional.
For example, in the above computation, the timestamps are two dimensional, the first is \emph{graph-version} and the second is \emph{Bellman-Ford iteration}, which we will refer to it as \emph{IFE iteration}, represented as a $\langle G_k, i\rangle$ pair. Collections, e.g., $D$, can change for two separate reasons: (1) changes in the graph ($E$), such as inserting an edge, or (2) changes in distances ($D$) during the computation of IFE iterations. 

More generally, for each data collection $C$, let $C_t$ represent the contents of $C$ at a particular timestamp $t$, and let $\delta C_t$ be the {\em difference set} that stores the ``difference tuples'' (differences for short) for $C$ at $t$. Differences
are extended tuples with + or - multiplicities. For base data collections, such as $E$,
+/- indicate external insertions or deletions to them. For intermediate
data collections,
these may not have as clear an interpretation.
Instead, the + or -'s are assigned to tuples to ensure
that summing all the $\delta C_t$ 
prior to a particular timestamp $t$ gives exactly $C_t$.
Sum of two difference sets adds the multiplicities 
for the differences with the same tuple values and if a sum equals 0, then the tuple is removed 
from the collection. 
Consider an operator with one input and one output collections, $I$ and $O$, respectively. 
DC ensures that for each collection and operator the following equations hold: 

\begin{equation}
I_t = \sum_{s \le t} \delta I_s ~\Rightarrow~  \delta I_t = I_t -  \sum_{s < t} \delta I_s
\label{eq1}
\end{equation}

\begin{equation}
O_t = Op(\sum_{s \le t} \delta I_s) ~\Rightarrow~ \delta O_t = Op(\sum_{s \le t} \delta I_s) - \sum_{s < t} \delta O_s
\label{eq2}
\end{equation}

DC uses Equations~\ref{eq1} and ~\ref{eq2} to compute the differences
to store in $\delta I_t$ and $\delta O_t$ for each timestamp.
Then, DC uses these difference sets to {\em reassemble} correct contents of $I_t$
and $O_t$ at each timestamp when needed during its maintenance procedure (explained momentarily).

\begin{figure}
	\centering
	\includegraphics[width=0.3\textwidth]{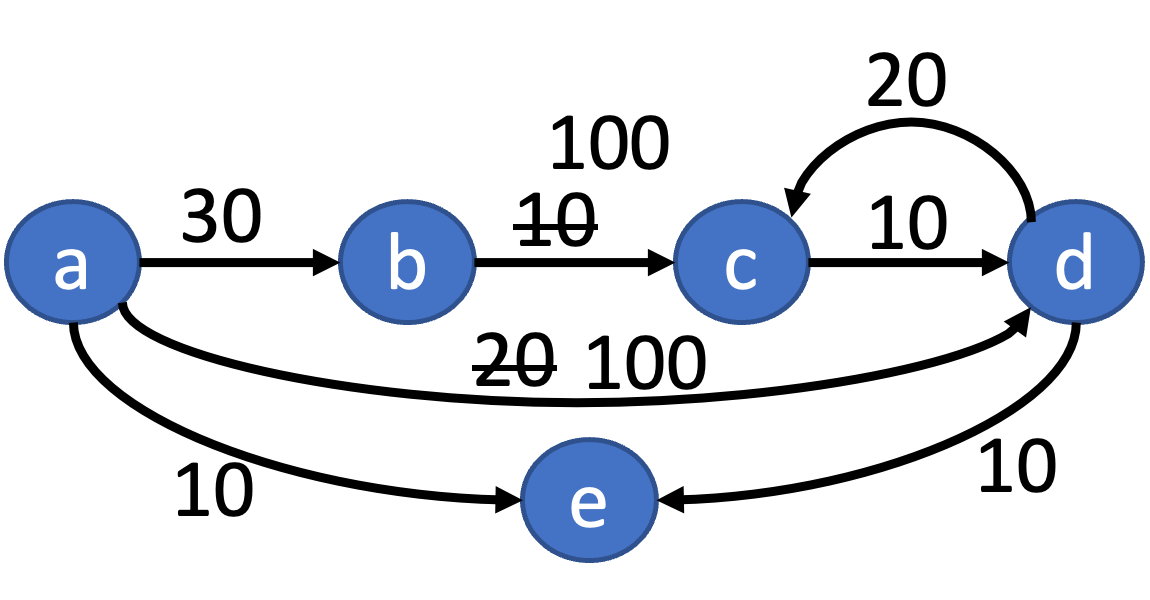}
	\vspace{-10pt}
	\caption{A dynamic graph with two updates: (i) $a$$\rightarrow$$d$ from 20 to 100 in $G_1$; and (ii) $b$$\rightarrow$$c$ changes from 10 to 100 in $G_2$.}
	\label{fig:running-ex-input}
	\vspace{-20pt}
\end{figure}

Suppose a system has maintained the Bellman-Ford dataflow differentially for $k$ many updates to its base collection $E$; that is, the system has computed the differences for each base or intermediate collection for timestamps $\langle G_0, 0 \rangle, \ldots,$  $\langle G_k, max \rangle$, where $max$ is the maximum number of iterations that the dataflow ran on any of $G_0, ..., G_k$. Given a new, $k+1$'st set of updates to the base collections, DC maintains the dataflow's computation by computing a new set of differences for collections at some of the timestamps $t = \langle G_{k+1}, i \rangle~|~i \in \{0...max\}$ by rerunning some of the operators at these timestamps. 
If on $G_{k+1}$, the Bellman-Ford dataflow computation requires more than $max$ iterations to converge, 
then the system generates difference sets for timestamps $\langle G_{k+1}, i\rangle~|~i > max$.

We next explain DC's maintenance procedure. Suppose that the operators work on partitions of collections.
In our example, the partitioning of the collections would be by vertex IDs and each operator would perform some computation per a vertex ID. 
Let $C^v_t$ indicate the contents of $C_t$'s partition for key $v$.
DC reruns an operator $Op$ at different timestamp $\tau$ according two rules:

\begin{squishedlist}
\item {\em Direct rerunning rule: }
if $Op$'s input $I$ has a difference at $\tau$ for a particular key $v$, i.e., $\delta I^v_{\tau}$ is non-empty, 
DC reruns $Op$ (on key $v$) at timestamp $\tau$. That is DC reassembles
$I^v_{\tau} = \sum_{t \le \tau} \delta I^v_{t}$ and executes $Op$ on $I^v_{\tau}$, which computes
a new $O^v_{\tau}$. Then, DC computes the difference set $\delta O^v_{\tau}$ as $\delta O^v_{\tau}=O^v_{\tau} - \sum_{t < \tau} \delta O^v_{t}$. 

\item {\em Upper bound rule:} 
For correctness, $Op$ may need to be executed on later timestamps than $\tau$ for $v$ even
if there is no immediate differences in $I$ at those timestamps. Specifically,
DC finds every timestamps $t_f \nless \tau$ in which $Op$'s input has differences for key $v$ 
and reruns $Op$ on timestamps that are least upper bounds of such $t_f$ and $\tau$. 
\end{squishedlist}

Importantly, if no difference is detected to vertex $v$'s partitions of inputs of an operator for timestamps from $\langle G_{k+1}, 0 \rangle$ to $\langle G_{k+1}, max \rangle$, no operator needs to rerun on $v$. For many dataflow computations, the effects of many updates in graphs can be localized to small neighbourhoods, and DC automatically detects the vertices in this neighbourhood on which operators need to rerun. 
\iftechreport
As an example, Table~\ref{fullTrace} shows 
 the full difference trace for each collection in the IFE
dataflow implementing Bellman-Ford algorithm in the example dynamic graph in Figure~\ref{fig:running-ex-input} that has 
two updates: (i) an update on (a, d) from 20 to 100, at timestamp $<$$G_1, 0$$>$; and (ii) an update on (b, c) edge from
10 to 100 at timestamp $<$$G_2, 0$$>$. These updates are modeled as differences in collection $E$ at these timestamps.
\fi
Reference~\cite{foundations} formally proves that applying this simple rule to decide which operators to rerun correctly maintains any dataflow computation.

\iftechreport
\begin{table}
    \centering
    \caption{Full trace of differences in the SPSP example from Figure~\ref{fig:running-ex-input}.}
    \setlength\tabcolsep{1.0pt}
    \setlength\extrarowheight{1pt}
    \small
    \begin{tabular}{@{\hskip-8pt}c|@{}c@{}|c|L{3.5cm}|L{1.5cm}|L{1.7cm}|} 
    \multicolumn{6}{c}{
    \begin{tikzpicture}
    \bfseries
    \draw[->, line width=0.2mm] (-2,0) --node  [above] {\small{Graph Updates}} (2,0);
    \end{tikzpicture}
    } \\\cline{2-6}
    &&&\centering{\textbf{$G_0$}} &  \centering{\textbf{$G_1$}} & \multicolumn{1}{|>{\centering\arraybackslash}m{1.5cm}|}{\textbf{$G_2$}} \\\cline{2-6}
    \multirow{12}{*}{
    \vspace{-2.5cm}
    \begin{tikzpicture}
    \bfseries
    \centering
    \draw[<-, line width=0.2mm] (0,-2) --node [midway, above, sloped] {\small{\IFE\ iterations}} (0,2);
    \end{tikzpicture}
    }
    &\multirow{3}{*}{\textbf{0}} & \textcolor{red}{\textbf{$\delta E$}} 
                        &$+(a, b, 30),$ $+ (b, c, 10)$,
                        $+(c, d, 10), +(a, d, 20)$,
                        $+(d, e, 10), +(a, e, 10)$,
                        $+(d, c, 20)$
                        &\textcolor{red}{$-(a, d, 20)$}, $+(a, d, 100)$ 
                        &\textcolor{red}{$-(b, c, 10)$}, 
                        $+(b, c, 100)$ \\\cline{3-6}
                    &   & \textcolor{blue}{$\delta J$} & $+(a,0), + (b, \infty),+(c, \infty),$ $+(d, \infty),+(e, \infty)$& $\varnothing$& $\varnothing$ \\\cline{3-6}
                    &  &\textcolor{teal}{\textbf{$\delta D$}} & $+(a,0), + (b, \infty),+(c, \infty),$ $+(d, \infty),+(e, \infty)$& $\varnothing$& $\varnothing$ \\
                    \cline{2-6}
                    \cline{2-6}
    &\multirow{3}{*}{\textbf{1}} & \textcolor{red}{\textbf{$\delta E$}} 
                        & $\varnothing$
                        &$\varnothing$ 
                        &$\varnothing$ \\\cline{3-6}
                        & & \textcolor{blue}{$\delta J$} 
                        & $+(b, 30)$, $+(d, 20)$, $+(e, 10)$
                        & \textcolor{red}{$-(d, 20)$}, $+(d, 100)$
                        & $\varnothing$ \\\cline{3-6}
                        & &\textcolor{teal}{\textbf{$\delta D$}} 
                        & \textcolor{red}{$-(b, \infty)$}$, +(b, 30)$,
                        \textcolor{red}{$-(d, \infty)$}, $+(d, 10)$,
                        \textcolor{red}{$-(e, \infty)$}, $+(e, 10)$
                        & \textcolor{red}{$-(d, 20)$}, $+(d, 100)$
                        & $\varnothing$ \\
                        \cline{2-6} 
                        \cline{2-6}
    &\multirow{3}{*}{\textbf{2}} & \textcolor{red}{\textbf{$\delta E$}} 
                        & $\varnothing$
                        & $\varnothing$ 
                        & $\varnothing$ \\\cline{3-6}
                        & & \textcolor{blue}{$\delta J$} 
                        & $+(c, 40)$, $+(c, 40)$, $+(e, 30)$
                        & \textcolor{red}{$-(c, 40)$}, $+(c, 120)$,
                        \textcolor{red}{$-(e, 30)$}, $+(e, 110)$
                        & \textcolor{red}{$-(c, 40)$}, $+(c, 130)$ \\\cline{3-6}
                        & &\textcolor{teal}{\textbf{$\delta D$}} 
                        & \textcolor{red}{$-(c, \infty)$}, $+(c, 40)$
                        & $\varnothing$
                        & \textcolor{red}{$-(c, 40)$}, $+(c, 120)$ \\
                        \cline{2-6} 
                        \cline{2-6}
    &\multirow{3}{*}{\textbf{3}} & \textcolor{red}{\textbf{$\delta E$}} 
                        & $\varnothing$
                        &$\varnothing$ 
                        &$\varnothing$ \\\cline{3-6}
                        & & \textcolor{blue}{$\delta J$} 
                        & $+(d, 50)$
                        & $\varnothing$
                        & \textcolor{red}{$-(d, 50)$}, $+(d, 130)$ \\\cline{3-6}
                        & &\textcolor{teal}{\textbf{$\delta D$}} 
                        & $\varnothing$ & \textcolor{red}{$-(d, 100)$}, $+(d, 50)$ & \textcolor{red}{$-(d, 50)$}, $+(d, 100)$ \\
                        \cline{2-6}
                        \cline{2-6}
    &\multirow{3}{*}{\textbf{4}} & \textcolor{red}{\textbf{$\delta E$}} 
                        & $\varnothing$
                        &$\varnothing$ 
                        &$\varnothing$ \\\cline{3-6}
                        & & \textcolor{blue}{$\delta J$} 
                        & $\varnothing$
                        & \textcolor{red}{$-(c, 120)$}, $+(c, 70)$,
                        \textcolor{red}{$-(e, 110)$}, $+(e, 60)$
                        & \textcolor{red}{$-(c, 70)$}, $+(c, 120)$,
                        \textcolor{red}{$-(e, 60)$}, $+(e, 110)$ \\\cline{3-6}
                        & &\textcolor{teal}{\textbf{$\delta D$}} 
                        & $\varnothing$ & $\varnothing$ & $\varnothing$ \\\cline{2-6}
    \end{tabular}
    \label{fullTrace}
    \end{table}
\fi

%% file: complete-dropping.tex
\section{Complete Difference Dropping: Join-On-Demand}
\label{sec:complete-drop}

When maintaining IFE with DC, the memory overheads of storing the difference sets for the 
output of the \texttt{Join} operator (\texttt{J})  is generally much larger 
than those for  the output of the following aggregation operator (\texttt{D}). 
Consider the IFE implementation of SPSP, where edges have weights and vertex states represent shortest distances to a source vertex. Suppose at a particular iteration $i$ of the IFE at a specific graph version $G_k$, 
a vertex $v$'s state is $(v, d_v)$ and $v$ has $deg(v)$ many
outgoing edges, e.g., $(u_1, w_1), \ldots (u_{deg(v)}, w_{deg(v)})$. Then to simulate $v$ 
propagating possible new shortest distances to its outgoing neighbours, \texttt{J}
would contain $deg(v)$ many tuples at timestamp $\langle G_k, i \rangle$:
$(u_1, d_v + w_1)$, ..., $(u_{deg(v)}, d_v + w_{deg(v)})$. 
Similarly, the partition \texttt{J}$^u$ of \texttt{J} contains
one tuple for each of $u$'s incoming neighbours.
When maintaining \IFE\ differentially, 
\texttt{J}'s size is commensurate with the number of edges
in $G$, which can be much larger than \texttt{D}, whose size is
commensurate with the number of vertices in $G$.

\begin{example}
\label{ex:deltaj}
Observe that in Table~\ref{fullTrace}, $\delta D$ has two differences  for vertex $d$ at timestamp $\langle G_1,1 \rangle$, 
$-(d, 20)$ and $+(d, 100)$. These changes  lead to four differences in $\delta J$ 
because $d$ has two outgoing edges, one to $c$ and the other to $e$.

\end{example}

The goal of \JOD\ is to avoid storing any difference sets  for \texttt{J}, i.e., to completely drop
$\delta J$, and regenerate $J^u$ for any $u$ on demand 
when DC requires running the aggregation operator (in our example \texttt{Min}) on $u$ at a particular
timestamp. We first describe an unoptimized
version of \JOD, then describe an optimization called {\em eager merging} that reduces
the timestamps to regenerate $J^u$, which is the optimized \JOD\ we 
have implemented.

\subsection{\JOD}
\label{subsec:jod}
Recall that DC reruns \texttt{Min} 
on a vertex $u$ at timestamp $t = \langle G_{k+1}, i \rangle$ 
if (1) $\delta D^u_{t}$ or
$\delta J^u_{t}$ are non-empty (direct rule);
or (2) $t$ is an upper bound of
$\tau_1$ and $\tau_2$ that satisfy the following conditions (upper bound rule):
(i) $\tau_1 \in \Tau_1 = \{ \langle G_{k+1}, i' \rangle |  i' < i \}$ and 
$\delta D^u_{\tau_1}$ and/or $\delta J^u_{\tau_1}$ are non-empty; and (ii)
$\tau_2 \in \Tau_2 = \{ \langle G_{k'}, i \rangle | k' < k+1\}$ and 
$\delta D^u_{\tau_2}$ and/or $\delta J^u_{\tau_2}$ are non-empty.
If
 $\delta J$ are dropped, how can we correctly
 decide when to rerun \texttt{Min} and recompute the needed dropped $\delta J$ 
 for these reruns to ensure we correctly differentially maintain IFE?
 DC$^{\text{JOD}}$ is our modified version of \DC\ maintenance
subroutine that has this guarantee, which works as follows.
In the below description, when \texttt{Min} is rerun on $u$ at timestamp $t$,    $J^v_t$ is constructed 
by inspecting for each incoming neighbour $w$ of $u$, 
 $D^w_{t}$ and $E^w_t$ and performing the join. 
 Note that we do not drop the differences related to \texttt{D} and \texttt{E}. 

\noindent DC$^{\text{JOD}}$:
\begin{squishedlist}
\item \em{$\delta E$ Direct Rule}: For each $(u, v, l, p, +/-) \in \delta E_{k+1}$, since there is a difference 
in $\delta E^u_{\langle G_{k+1}, 0 \rangle}$, there is also a difference
in $\delta J^v_{\langle G_{k+1}, 0 \rangle}$. So
we rerun \texttt{Min} on $v$ in $\langle G_{k+1}, 0 \rangle$ (direct rule). 
 
\item \em{$\delta D$ Direct Rule}: Each time \texttt{Min} reruns on $u$ at a timestamp $\langle G_{k+1}, i\rangle$,
we check if it generates a difference for $\delta D^u_{\langle G_{k+1}, i+1\rangle}$.
If so, this implies there is a 
difference in $\delta J^v_{\langle G_{k+1}, i+1\rangle}$ for each outgoing 
neighbour $v$ of $u$. Therefore we schedule \texttt{Min} on $v$ at timestamp
$\langle G_{k+1}, i+1\rangle$ (direct rule). 

\item {\em Upper Bound Rule:}
Each time we schedule to rerun \texttt{Min} on a vertex $v$, either by  
$\delta E$ or $\delta d$ Direct Rule at timestamp $\langle G_{k+1}, i+1\rangle$,
by the upper bound rule, we schedule to
rerun \texttt{Min} on $v$ at timestamp $\langle G_{k+1}, j\rangle$ s.t. $j > i+1$ if either 
of these two conditions are satisfied: (i) there is a non-empty
$\delta D^v_{\langle G_h, j \rangle}$ s.t $h < k + 1$; and (ii) there is
an incoming neighbour $w$ of $v$ with a non-empty $\delta D^w_{\langle G_h, j \rangle}$ s.t., 
$h < k +1$.
\end{squishedlist}

\iftechreport
Our next theorem proves that DC$^{\text{JOD}}$ correctly
maintains the IFE dataflow.

\begin{theorem}
\label{thm:dc-jod}
The subset of timestamps
that DC$^{\text{JOD}}$ re-computes \texttt{Min} on any key/vertex ID 
subsumes the timestamps that DC re-computes \texttt{Min} and 
correctly generates the same set of differences for \texttt{D}.
\end{theorem}

\begin{proof}
Our proof is very technical and by induction on timestamps $t$. We assume
for simplicity that there is a global $max$ iterations 
on which the \IFE\ computations runs. 
We take 
as the base cases $\langle G_0, 0 \rangle$
to $\langle G_0, max \rangle$, as in this case the behaviour
of  DC$^{\text{JOD}}$ simply follows $DC$, which 
further directly follows the computation that is performed when 
running static version of \IFE\ on an input graph. In this case the behaviour
is that \texttt{Min} runs on all vertices in $\langle G_0,0\rangle$ and
then for each vertex $v$ at $\langle G_0,0\rangle$ if one of its incoming
neighbour's states change. As induction hypothesis, we assume
that for each vertex $v$, DC$^{\text{JOD}}$ runs \texttt{Min} on a larger subset of
timestamps and correctly generates the same set of differences 
for the data collection \texttt{D} (and \texttt{E}, whose maintenance is independent
of the \JOD\ optimization) until $\langle G_k, max\rangle$. We will prove
that for each key $v$, if \DC\ runs \texttt{Min} in timestamp $\langle G_{k+1}, i\rangle$,
so does DC$^{\text{JOD}}$. Note that for $t_0=\langle G_{k+1}, 0\rangle$,
this is true because if $v$ is re-computed by DC in $t_0$
it is because there is change in $\delta J^v_{t_0}$, which
can only occur if there is an edge $(u, v)$ in $\delta E_{t_0}$,
i.e., one of $v$'s incoming edges must have had an edge update that 
triggered \DC\ to rerun \texttt{Join} which must trigger a difference
in $\delta J^v_{t_0}$. The first step of  DC$^{\text{JOD}}$
ensures that for each edge (u, v) in $\delta E_{t_0}$, \texttt{Min}
is re-computed for $v$ in $t_0$. We use this as a base case, we do another 
induction only on the second component of the timestamps in $\langle G_{k+1}, i\rangle$.
That is we assume from $\langle G_{k+1}, 0\rangle$ to $\langle G_{k+1}, i\rangle$,
the theorem is true and prove it for $\langle G_{k+1}, i+1\rangle$.

Now consider $t_{i+1} = \langle G_{k+1}, i+1\rangle$. We will prove the contrapositive
of the claim: if DC$^{\text{JOD}}$ does not re-compute \texttt{Min} on $u$,
then neither will DC. Note that DC$^{\text{JOD}}$ does not run on $u$ if two conditions
are satisfied: (1) none of $u$'s incoming neighbours $v$ had
a non-empty $\delta D^u_{\langle G_{k+1}, i\rangle}$. If there was we would 
schedule $u$ to re-compute, which means \DC\ cannot execute 
\texttt{Min} on $v$ as an application of the direct rule. (2) none of 
$u$'s incoming neighbours $w_1$, ..., $w_2$ has
a non-empty $\delta D^w_{\langle G_h, j \rangle}$ s.t., 
$h < k +1$. What we next show is that if this condition is true,
then all such $\delta J^u_{\langle G_h, j \rangle}$
are non-empty, so \DC\ cannot have triggered the upper bound rule either.

Finally, to prove this we will do a third induction starting
from $\delta J^u_{\langle G_0, j \rangle}$ to
$\delta J^u_{\langle G_k, j \rangle}$. For the base case of 
$\delta J^u_{\langle G_0, j  \rangle}$, observe that:
 
$$\delta J^u_{\langle G_0, j \rangle} = J^u_{\langle G_0, j \rangle} - 
J^u_{\langle G_{0}, j-1 \rangle} = \emptyset$$
This is true because this is a 1 dimensional case and 
$J^u_{\langle G_0, j \rangle}$ and $J^u_{\langle G_0, j-1 \rangle}$
are computed buy the \texttt{Join}
operator using the same of incoming neighbour states (recall that our assumption
is that all $\delta D^w_{\langle G_0, j \rangle}$ are empty).
Suppose now that $\delta J^u_{\langle G_0, j \rangle}$ up to
$\delta J^u_{\langle G_h, j \rangle}$ are empty. We prove that 
$\delta J^u_{\langle G_{h+1}, j \rangle}$ is empty:

$$\delta J^u_{\langle G_{h+1}, j \rangle} = 
J^u_{\langle G_{h+1}, j \rangle} - J^u_{\langle G_{h+1}, j-1 \rangle} -
\sum_{\ell = 0, ..., h} \delta J^u_{\langle G_{\ell}, j \rangle}
 $$ 
 By induction hypothesis the last summation is non-empty.
 Note further that 
$J^u_{\langle G_{h+1}, j \rangle}$ and $J^u_{\langle G_{h+1}, j-1 \rangle}$
are the same set because they are computed by the \texttt{Join}
operator using the same of incoming neighbour states for $u$, 
as we are assuming that all $\delta D^w_{\langle G_h, j \rangle}$ are empty for 
each incoming neighbour $w$ of $u$, completing the proof.
\end{proof}

Note that there can be timestamps in which DC$^{\text{JOD}}$
unnecessarily re-compute \texttt{Min}, but by the correctness argument
of DC in reference~\cite{foundations}, any timestamp that DC avoids
rerunning a computation, is guaranteed to produce
empty differences, so these spurious re-computations cannot affect 
the correctness of DC$^{\text{JOD}}$, so as a corollary of Theorem~\ref{thm:dc-jod},
we can conclude that DC$^{\text{JOD}}$ correctly maintains the \IFE\ dataflow. 
A simple example of spurious rerun is the simple case in the SPSP example
is when a vertex $u$ has two incoming edges, say from $w_1$ to $w_2$,
where suppose for purpose of demonstration $w_1$ and $w_2$ start with states 0
initially (so at a timestamp $\langle G_0, 0\rangle$), and suppose further that
edges $(w_1, u)$ has a weight of 10 and $(w_2, u)$ has a weight of 20,
so the $J^u_{\langle G_0, 1 \rangle}$ would contain $(u, 10, +)$ and $(u, 20, +)$.
Suppose in $G_1$, the weights of these edges are swapped. The original DC
 would not re-compute  \texttt{Min} at timestamp $\langle G_1, 1\rangle$ on
$u$, because there is no difference directly to $u$'s input (nor through an upper bound rule), 
where as DC$^{\text{JOD}}$ would. This is because DC$^{\text{JOD}}$ would immediately
schedule \texttt{Min} to execute on $u$ because $w_1$ (or $w_2$'s) states
changed at  $\langle G_1, 0 \rangle$ and $u$ is an outgoing neighbour of $w_1$.

\fi

\begin{example}
\label{ex:jod}
We next demonstrate applications of JOD's rerunning rules on our running example.
Consider the first update in our running example at timestamp $\langle G_1, 0 \rangle$,
which updates the weight of edge (a, d) from 20 to 100. By the $\delta E$ Direct Rule of 
JOD, we rerun \texttt{Min} on $d$ at timestamp $\langle G_1, 0 \rangle$.
Further by JOD's Upper Bound Rule, we also schedule to run $d$
at timestamp  $\langle G_1, 2 \rangle$ because $\delta D^c_{\langle G_0, 2 \rangle}$
is non-empty and $c$ is an incoming neighbour of $d$ (condition (ii)).
Note that rerunning \texttt{Min} on $d$ at timestamp $\langle G_1, 0 \rangle$
 creates a difference for $\delta D^d_{\langle G_1, 1\rangle}$.
By the $\delta D$ Direct Rule, we further schedule to rerun \texttt{Mi}n on $c$ and $e$, which
are the outgoing neighbours of $d$, at timestamp 
$\langle G_1, 1\rangle$. 
\end{example}

\vspace{-10pt}
\subsection{Eager-Merging}
\label{subsec:em}
The naive implementation of \JOD\ can be expensive because  
the number of possible timestamps $\langle G_{h}, i \rangle$ to inspect, where $h < k+1$,
can grow unboundedly large as batches of edge updates continue to arrive. 
The eager merging optimization we describe next, which extends
a periodic merging optimization of the \DD\ system (explained momentarily), reduces
the number of these timestamps.

Consider the point
at which a new set of updates to graph version $G_{k}$ has
arrived and the system 
has finished maintaining the computation for $G_{k}$.
So there are $k \times max$ many different timestamps
in the computation so far. Let us think of these
timestamps in a 2D grid with columns as 
graph version indices and rows as IFE iterations 
as in Table~\ref{fullTrace}. 

Observe that as more updates arrive to the system,
the timestamps will increase in the 
graph version dimension to $G_{k+1}$, $G_{k+2}$, etc, so more columns will be added to this grid.
Consider reassembling the contents of some collection $C$ 
at timestamp $\langle G_{k+1}, 0\rangle$. To do so,
\DD\ has to sum the differences 
in $\delta C_{Row_0} = \{\delta C_{\langle G_0, 0\rangle}$, ..., 
$\delta C_{\langle G_{k}, 0\rangle}\}$. 
To reassemble $C$ at  timestamp $\langle G_{k+1}, 1\rangle$,
\DD\ has to sum the difference sets in $\delta C_{Row_0}$ 
and $\delta C_{Row_1} = \{\delta C_{\langle G_0, 1\rangle}$, ..., 
$\delta C_{\langle G_k, 1\rangle}\}$, etc.
Observe that once the (k+1)'st graph updates have arrived, the system
will never have to re-execute
an operator at timestamps 
$\langle G_h, i \rangle$ where $h < k+1$. 
Instead of computing 
$\delta C_{Row_j}$ multiple
times for each possible $\{ C_{\langle G_{k+1}, j\rangle} | j > i \}$,
the original \DD\ periodically unions the individual difference
sets in $\delta C_{Row_j}$ into a single difference set  $\delta C_{\langle G_{k+1}, j\rangle}$. 
This allows \DD\ to reassemble collections faster and store the difference
sets more compactly.

Instead of periodic merging, we eagerly merge the differences along the $\text{graph-version}$ dimension as we run DC's maintenance procedure for 
$\langle G_{k+1}, 0\rangle$ to $\langle G_{k+1}, max\rangle$.
That is, as soon as DC finishes maintaining $\langle G_{k+1}, i\rangle$,
we merge the difference sets for \texttt{D} for timestamps
$\langle G_{k}, i\rangle$ and $\langle G_{k+1}, i\rangle$.
This guarantees that for any vertex, we only need to keep one-dimensional timestamps, i.e., only for $\text{\em{IFE iteration}}$. 
Table~\ref{fig:earlyMerge} shows the states of the differences stored in the system with  
eagerly merging differences and the DC algorithm is in the process of maintaining the computation
at timestamp $\langle G_2, 2 \rangle$.
Differences at grey cells have been merged to the right most cell on the row.
 In presence of eager merging,
whenever \JOD\ needs to investigate if 
 $\delta D_{\langle G_h, i \rangle}~|~h < k +1$ is non-empty for any vertex,
 we only need to inspect timestamps with $h = k$.
  
We end this subsection with a discussion of another benefit of eager merging. Eager merging allows dropping all differences with negative multiplicities in  the difference sets for \texttt{D}. 
This is because in the algorithms we consider,
vertices take one unique state at each iteration of IFE. Therefore in one-dimensional timestamps, 
the change in the state of a vertex from $s$ to $s'$ at iteration $i$,
is always represented with two differences: (i) one with positive multiplicity with $s'$; and (ii) one with negative multiplicity for $s$. 
\iftechreport
For example, readers can check that once DC$^{JOD}$ with eager merging finishes maintaining
the computation for all timestamps for graph version $G_2$, 
the distances stored for vertex $d$ will be  
 \{$(1,100, +)$,$(3,100, -)$,$(3,50, +)$\}, where the first values in these tuples are the timestamps, which
 we now represent only with IFE iteration number.
These differences can be stored as  \{$(1,100)$, $(3,50)$\}, and the 
 $(3, 100, -)$ would be implied. 
 \fi
 In absence of negative multiplicities, we can also avoid
 doing any summations when computing the state of a vertex at timestamp $i$, i.e., $D^v_i$. Instead
 we can find the latest iteration $i^* \le i$ in which vertex $v$ has a (positive) difference
 and return it.


\begin{table}
    \centering
    \caption{Differences in \texttt{D} on our running example \textbf{with eager-merging} when maintaining the computation for  $\langle G_2, 2 \rangle$.}
    \vspace{-10pt}
    \setlength\tabcolsep{1.0pt}
    \setlength\extrarowheight{1pt}
    \small
    \begin{tabular}{@{\hskip-8pt}c|@{}c@{}|L{1cm}|L{2cm}|L{4cm}|} 
    \multicolumn{5}{c}{
    \begin{tikzpicture}
    \bfseries
    \draw[->, line width=0.2mm] (-2,0) --node  [above] {\small{Graph Updates}} (2,0);
    \end{tikzpicture}
    } 
    \\\cline{2-5}
    &&\centering{\textbf{$G_0$}} &  \centering{\textbf{$G_1$}} & \multicolumn{1}{|>{\centering\arraybackslash}m{1.5cm}|}{\textbf{$G_2$}} 
    \\\cline{2-5}
    \multirow{4}{*}{
    \begin{tikzpicture}
    \bfseries
    \centering
    \draw[<-, line width=0.2mm] (0,-1) --node [midway, above, sloped] {\small{IEF iterations}} (0,0);
    \end{tikzpicture}
    }
    &{\textbf{0}} &\cellcolor{gray!75}  & \cellcolor{gray!75}  & $+(a,0), +(b, \infty), +(c, \infty),$ $+(d, \infty), +(e, \infty)$ 
    \\\cline{3-5}
    &{\textbf{1}} & \cellcolor{gray!75} & \cellcolor{gray!75} &  $+(b, 30), +(d,100), +(e,10)$
    \\\cline{3-5}
    &{\textbf{2}} & \cellcolor{gray!75} & $+(c, 40)$ & \cellcolor{red!50} 
    \\\cline{3-5}
    &{\textbf{3}} & \cellcolor{gray!75} &  $+(d,50)$  & 
    \\\cline{3-5}
    \end{tabular}
    \vspace{4pt}
    \label{fig:earlyMerge}
    \vspace{-20pt}
    \end{table}

%% file: partial-dropping.tex
\section{Partial Difference Dropping}
\label{sec:partial}

We next investigate optimizations that partially drop the differences in $D$.
When we apply \JOD, $D$ is the only data collection 
for which we store differences, except for the original edges in the graph. 
Partial dropping the differences in $D$ allows trading off scalability with query performance.
Specifically, the memory overhead to store $D$ decreases, yet it also decreases performance 
because when DC needs to reassemble the contents of $D$ at a timestamp $t$,
the dropped differences need to be recomputed. In this section, we will describe optimizations 
with different scalability/performance tradeoffs.
Throughout this section we assume running DC$^{JOD}$ with eager merging and use
single dimensional timestamps to refer to data collections, such as $D_i$, instead of $D_{\langle G_k, i \rangle}$.

A partial dropping optimization has two key components:
\begin{squishedlist}
\item {\em Dropped Difference Maintenance}: When DC$^{JOD}$ accesses 
$D^v_i$, the system needs to identify if a difference was dropped with key/vertex ID $v$ at timestamp $i$. Therefore, the system needs to maintain the dropped vertex ID-timestamp pair information. 

\item  {\em Selecting the Differences to Drop:} The system also needs to decide which differences to drop and which ones to keep.
\end{squishedlist}
\noindent We describe alternative approaches to both components.

A third important decision is 
to choose how many differences to drop given a memory constraint. At a high-level, the answer to this question
is clear:  drop as little as possible without violating the memory constraint. In practice however estimating this 
amount may be challenging because each update to the graph changes the amount of differences needed 
to maintain registered queries. Further, a system needs to estimate and plan for
newly registered or deregistered queries. In such dynamic scenarios, systems can adopt
adaptive techniques that determine how many differences to drop from each query by observing 
the stored differences. This is a future topic for a  rigorous within study. Within the context 
of this paper, we will assume a user-define probability $p$ 
that drops each difference with probability $p$ (see Section~\ref{sec:selectingDiffs}). 

\vspace{-7pt}
\subsection{Dropped Difference Maintenance}
\label{sec:PartialCorrectness}

One natural approach to maintaining the dropped vertex ID-timesta-mp pairs (VT pairs for short) is to store them 
explicitly in a separate data structure \texttt{DroppedVT}. 
We  discuss two possible designs for this data structure. We first present a
straightforward deterministic data structure,
and discuss its scalability bottlenecks. Then, we propose a probabilistic data structure, which can address this
scalability bottleneck but possibly leading to spurious re-computations of undropped differences. 
In our evaluations, we show that, despite this possible performance disadvantage,
our probabilistic approach can still be more performant as it can drop fewer differences than our
deterministic approach under limited memory settings.

\subsubsection{Deterministic Difference Maintenance (\textsc{Det-Drop})}
\label{sec:optimization-det}
\textsc{Det-Drop} uses a hash table to implement
 \texttt{DroppedVT}. 
During $DC^{JOD}$, when $D^v_i$ is needed, we perform the following \texttt{AccessD$^v_i$WithDrops} procedure. Before 
we describe our procedure, recall from Section~\ref{subsec:em} that we do not store 
differences with negative multiplicities for $D$ when
we eagerly merge differences, so we do not need to do any summation to compute $D^v_i$. We only need 
to find and return the latest iteration $i^* \le i$ for which there is a difference for $v$.

\noindent \texttt{AccessD$^v_i$WithDrops}:
\begin{squishedenumerate}
\item[1.] Let $\delta D$ be the index that stores the difference sets for $D$. We check $\delta D$ 
for the latest iteration $g^* \le i$, if any, for which the system has stored a difference for $v$.
\item[2.] Check \texttt{DroppedVT} for the latest iteration $d^* \le i$, if any, for which the system has dropped a difference 
for vertex $v$.
\item[3.] If a $d^* > g^*$ exists, recompute the stored difference at $d^*$ and return this value. 
Otherwise, return the value at $\delta D$ at $g^*$.
\end{squishedenumerate}
\noindent 
Note that to recompute a dropped difference at timestamp $d^*$ in step $3$
we rerun the aggregation operation, e.g., \texttt{Min}, for vertex $v$ at 
iteration $d^*-1$. This procedure is similar to how we rerun \texttt{Min} operator
for vertices at different timestamps as part of the $DC^{JOD}$ procedure.
Then using $J^v_{d^*-1}$ and $D^v_{d^*-1}$ we rerun \texttt{Min} and compute $D^v_{d^*}$. However when
we access  $D^v_{d^*-1}$, we recursively call  \texttt{AccessD$^v_i$WithDrops}, as there
may be dropped differences for $v$ or one of its incoming neighbours $w$ at timestamp $d^*-1$. Therefore,
this may lead to further recomputations, which may further cascade.

\begin{example}
\label{ex:partial}
Consider the running example. Suppose that  after the first update, the system
decides to drop the  difference $+(b, 30)$ at iteration 1.\
Consider now the arrival of the second update where the weight of $(b, c)$ changes from 10 to 100. 
To maintain the computation differentially, \texttt{Min} is rerun on 
$c$ at $\langle G_2, 1 \rangle$ (due to $\delta E$ Direct Rule) and 
then due to the Upper Bound Rule on every timestamp in which $c$ has a difference. 
$c$ already has a difference that is not dropped at iteration 2, so $c$is scheduled  to rerun at iteration 2.
We further check if $c$ has any dropped differences at iterations 3 and 4. Since it does not,
we do not schedule $c$ to rerun at these differences.
Then when rerunning $c$ at 2, we need both $b$'s and $d$'s distances at iteration 1 and check if they have
any differences. We see that $d$ has stored difference but $b$ does not, so we check if $b$
 has a dropped difference at 1. Since it does, we recompute that difference by rerunning $b$ at iteration 1.
\end{example}
\vspace{-5pt}

Explicitly keeping track of all dropped VT pairs 
requires keeping additional state that is proportional to the number of differences that are dropped, which limits its scalability. 
Note that a difference is simply a triple that consists of a VT pair plus a vertex state (e.g., distance).
Suppose we need $d$ bytes to store VT pairs and $s$ bytes to store the actual state in a difference. Then, for each dropped $d+s$ bytes, we have to keep $d$ bytes in \texttt{DroppedVT}. This means that even if we partially drop $100\%$ of differences, there is a hard limit of $\frac{d}{d+s}$ on the scalability benefits we can obtain from deterministically
dropping differences. Our next optimization overcomes this limitation by using a probabilistic data structure.

\begin{figure}[t]
	\centering
\resizebox{\columnwidth}{!}{%
\begin{tikzpicture}
\draw (0, 0) node[inner sep=0] {\includegraphics[width=5in]{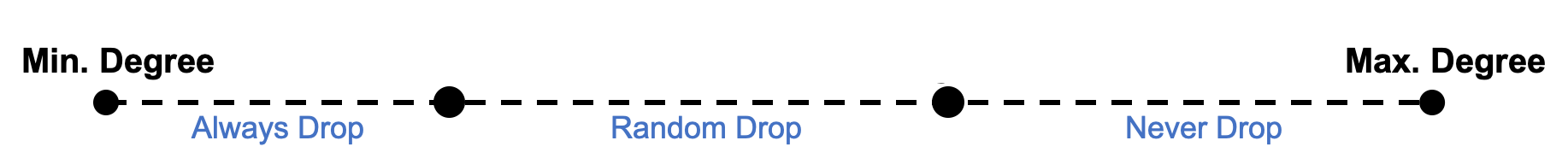}};
\draw (-2.7, 0.3) node {\color{black} \Large $\tau_{min}$};
\draw (1.3, 0.3) node {\color{black} \Large $\tau_{max}$};
\end{tikzpicture}
}
    \vspace{-15pt}
	\caption{Degree-Drop Strategy for dropping \diffs}
	\label{fig:selectiveExplained}
	\vspace{-10pt}
\end{figure}

\subsubsection{Probabilistic Difference Maintenance (\textsc{Prob-Drop}) }
\label{sec:optimization-prob}

\textsc{Prob-Drop} drops the entire difference, i.e., both the VT pair and the state, and 
uses a probabilistic data structure to maintain the dropped VT pairs. 
Probabilistic data structure, such as Bloom~\cite{bloom1970} or Cuckoo filters~\cite{cuckoo2014}, have the advantage that 
their sizes can remain much smaller than the amount of data they store. 
\textsc{Prob-Drop} requires a probabilistic data structure that never returns false negatives because
if a VT pair was dropped and the structure returns false when queried, we may ignore this
difference and reassemble incorrect states for vertices during $DC^{JOD}$. However, the structure can return
false positives, because false positives can only lead to unnecessarily recomputing a vertex state, but the recomputed vertex state will still be correct. 

We use a Bloom filter\footnote{We use the Bloom filter implementation from~\url{https://github.com/lemire/bloofi}}, into which we insert the dropped VT pairs. Using a Bloom filter requires minor modifications to the
\texttt{AccessD$^v_i$WithDrops} procedure from Section \ref{sec:optimization-det}. Specifically, in the second step, the procedure needs to check the Bloom filter for each potentially dropped difference at iteration $d \in (g^*, i]$ starting from $i$ to see if a VT
pair for $(v,d)$ was dropped. 
If the answer is negative, then processing moves to the next $d$ until
we arrive at $g^*$. In this case, the value from \texttt{D} for iteration $g^*$ (obtained from step 1) 
is the correct value of $v$ at iteration $i$. If the answer is positive for an iteration $d^* \in (g^*, i]$, then the 
value of pair $(v, d^*)$ is recomputed. 

In our evaluations, we show that \textsc{Prob-Drop} can increase the scalability of a GDBMS 
more than \textsc{Det-Drop} because its size does not grow as the system drops more \diffs. Furthermore, in some settings, 
the system does not need to drop as many \diffs\ in \textsc{Prob-Drop} as in \textsc{Det-Drop} to 
reach a certain scalability level (in our evaluations this is the number of concurrent queries).
\vspace{-7pt}
\subsection{Selecting the Differences To Drop}
\label{sec:selectingDiffs}

The second component of a partial difference dropping optimization is to decide which differences to drop. 
A baseline heuristic is to drop each difference uniformly at random. 
We next show a more optimized technique that use the degree information of vertices to select the differences to drop.

\subsubsection{Degree-based Difference Dropping}
\label{sec:degree-drop}

A GDBMS using DC to maintain continuously running recursive queries can exploit the fact that the 
dataset is a graph, therefore partitioning keys are vertex IDs. 
Intuitively, when executing the recursive algorithms we consider, high degree vertices are used frequently when computing the states of other vertices, i.e., they will be accessed more by DC when maintaining the input IFE dataflow. Therefore, dropping their differences can lead to frequent vertex state re-computations. Similarly, vertices with low-degrees are relatively less frequently accessed by DC. Based on this intuition, we implement a heuristic that takes two thresholds $\tau_{min}$ and $\tau_{max}$, for minimum and maximum degrees, respectively, and a probability parameter $p$. Then,
our heuristic performs the following for a difference with a VT pair $\langle$vertex, iteration$\rangle$ pair ($\langle v, i\rangle$) assuming that $deg(v)$ is the degree of vertex $v$ (Figure~\ref{fig:selectiveExplained}):

\begin{squishedlist}
\item If $deg(v) < \tau_{min}$, drop the difference.
\item If $deg(v) > \tau_{max}$, do not drop the difference.
\item Otherwise drop the difference with probability $p$.
\end{squishedlist}
\noindent We found empirically that setting  $\tau_{min}$ as 2 and $\tau_{max}$ as the top 80th degree percentile 
of the input graph is reasonable for the graphs we used in the experiments.
We note that more sophisticated properties, such as betweenness centrality of vertices, can also be used to decide the differences to drop. A practical advantage of using degrees is that, degree information is readily 
available in adjacency list indices, which are ubiquitously used in GDBMSs.


%% file: evaluation.tex
\section{Evaluation}
\label{sec:evaluation}

\subsection{Experimental Setup}

We run all experiments on a Linux server with $12$ cores and $32$ GB memory, unless mentioned otherwise. For each experiment, we report the total time, in single-threaded execution, needed to update the graph and the query answer after applying a batch of updates. For each dataset, we shuffle the edges, and split the dataset such that 90\% of the data is used as an initial graph, while the remaining 10\% models the dynamism in the graph consisting of the update to the graph. We use a default batch size of $1$, because differential computation is more suitable for near real-time dynamic graph updates than for infrequent updates. We evaluate the effects of batch size on
the performance of DC  
\iftechreport
in Appendix~\ref{sec:batchSizeImpact}.
\fi
We use insertion-only workloads in our main experiments. 
\iftechreport
In Appendix~\ref{exp:delete}
\fi
we also present experiments that use workloads with different amounts of deletions for our main experiments.

\subsubsection{Datasets}
\label{sec:dataset}
We use a combination of real and synthetic graphs summarized in Table~\ref{datasets}\footnote{Reported degrees are for the initial loaded graphs in the experiments.}. The four real graphs are \textbf{Skitter},  \textbf{LiveJournal}, \textbf{Patents}, and \textbf{Orkut}, all obtained from~\cite{snapDatasets}. \textbf{Skitter} represents an internet topology from several scattered sources to millions of destinations on the internet and its vertices are strongly connected. 
\textbf{LiveJournal} and~\textbf{Orkut} ~\cite{snapDatasets} represent social network interactions with a vertex degree distribution that follows power-law. 
\textbf{Patent} ~\cite{snapDatasets} represents a citation graph for all utility patents granted between 1975 and 1999.
In order to experiment with weighted SPSP queries, we created weighted versions of both graphs by adding a
random integer weight between $1$ and $10$ uniformly at random to each edge.
 \textbf{LDBC SNB}~\cite{LDBC-SNB} is a synthetic graph that models dynamic interactions in social network applications. This graph has edge labels that are used in RPQ queries. LDBC SNB includes several types of entities, such as persons or forums. Each edge has a label such as \textsf{Knows} or \textsf{ReplyOf}. We use a scale factor of $10$ that generates a graph of $7.2$M vertices and $77.6$M edges. 

\begin{table}
\footnotesize
	\centering
	\caption{Datasets}
	\vspace{-10pt}
	\begin{minipage}{\textwidth}
	\begin{tabular}{| l || c | c | c | c | c |}
		\hline
	    {\bf Name} & {\bf $|E|$} & {\bf $|V|$} & {\bf Max.} & {\bf Avg.} & {\bf Avg.} \\
	    & & & \textbf{Degree} & \textbf{Degree} & \textbf{In-Degree}\\
		\hline
		\hline
		LiveJournal (LJ)  & $69$M & $4.8$M & $4$K & $8.5$ & $14.2$\\
		\hline
		Skitter (SK) & $11$M & $1.7$M & $35$K & $8.2$ & $12.6$\\
		\hline
		Patents & $16.5$M & $3.8$M & $704$ & $2.3$  & $4.7$ \\
		\hline
		Orkut & $117.2$M & $3$ & $29.6$K & $17.7$ & $34.4$ \\
		\hline
		LDBC SNB & $77.6$M & $7.2$M & $20.8$K & $7.3$ & $9.8$ \\
		\hline
	\end{tabular}
	\end{minipage}
	\label{datasets}
	\vspace{-10pt}
\end{table}

\subsubsection{Workloads}

We use SPSP, K-hop, and several popular RPQ queries as our main query workloads.  We run SPSP and K-hop on the weighted and unweighted versions of the real datasets, respectively.  For SPSP query generation, we pick a random pair of vertices in the graph. For K-hop, we pick a random set of vertices and set the value of maximum hops $K=5$ to make it a 5-hop query. 

RPQ queries require edge labels, so this experiment is conducted only on the LDBC dataset. We use a set of RPQ templates used in real-world workloads as defined in reference~\cite{bonifati2019navigating} which were used to study streaming RPQ evaluation in reference~\cite{sigmod20_PacaciBO20}. 
\iftechreport
There are only two recursive relationships in LDBC SNB: \textsf{Knows} and \textsf{ReplyOf}. Recursive here refers to an edge label that can exist consecutively in an arbitrary path. Therefore, some templates that expect more than two recursive relationships cannot be used in LDBC SNB. 
\fi
We use the following RPQ query templates:
\begin{itemize}
	\item $Q_1 = a^*$ 
	\item $Q_2 = a \circ b^* $ 
	\item $Q_3 = a \circ b \circ c \circ d \circ e $ 
\end{itemize}
We used \textsf{Likes}, \textsf{Knows}, \textsf{ReplyOf}, and \textsf{hasCreator}, to construct queries from these templates in the LDBC SNB dataset.

SPSP, K-hop, and RPQs are queries that 
can be supported in high-level languages of GDBMSs. These are the main queries that motivate our work. However our optimizations
are applicable to other computations based on IFE. To demonstrate this, we 
implemented the differential versions of  standard weakly connected 
components (WCC) algorithm, which is based on iteratively propagating 
and keeping track of minimum vertex IDs, and PageRank (PR) (ran a fixed 10 number of iterations) in our setting.

\subsubsection{Baselines and Different GraphflowDB Configurations}
\label{sec:baselineDescribtion}

We implement our optimizations inside the continuous query processor (CQP) of GraphflowDB~\cite{grapflow}, which is a shared memory GDBMS. We extended the CQP of GraphflowDB to implement a baseline \DC\ and our optimizations to maintain the recursive queries we cover (see 
\iftechreport
Appendix~\ref{sec:implementation}
\fi
for the details of our implementation).
We call the GraphflowDB configurations for
 different configurations of \DC\ as: \VDC, \JOD, \DET, or \PROB. 

We compare our proposed optimizations with three baselines: \DD, \SC, and \DC. \DD~ is an implementation of our workloads in the Differential Dataflow system~\cite{differentialGithub}, which is the reference implementation of differential computation.
\SC\ is a baseline extension of GraphflowDB's CQP to support our queries by simply executing each query from scratch after every batch of changes. \SC\ represents a baseline GDBMS's performance that does not support continuous queries. 
We use an \IFE-like label propagation algorithm for K-hop queries and RPQs.
We note that this algorithm is identical to what is referred to as the ``incremental'' fixed point algorithm in the original
Differential Dataflow paper~\cite{mcsherry:ddf} (see Figure 1 in the reference). This term is used to indicate that 
only the vertices whose values are updated in a particular iteration propagate their labels in that iteration (as opposed to all vertices).  

\VDC~ is the vanilla differential computation implementation in GraphflowDB. 
The difference between \VDC~ and \DD~ is that the former is our single machine implementation using Java while the latter is a distributed system implemented in Rust.  
In Section~\ref{exp:baseline}, we verify that \VDC\ behaves similar to \DD\ (and even outperforms it in terms of runtime); therefore, we use \VDC~ as a suitable baseline for our optimizations that is implemented inside the same GDBMS.
\VDC\ ingests and stores the input graph in the same way, uses similar data structure to store the differences,
and the same programming language as the following GraphflowDB configurations:

\begin{enumerate}
	\item \JOD: The \DC\ version that implements join-on-demand optimization from Section~\ref{sec:complete-drop};
	\item \DET: Integrates deterministic partial dropping optimization on top of \JOD\ as discussed in Section~\ref{sec:optimization-det};
	\item \PROB: Integrates probabilistic partial dropping optimization on top of \JOD\ as discussed in Section~\ref{sec:optimization-prob}.
\end{enumerate}
We also evaluate different versions of \DET\ and \PROB\ to evaluate our degree-based difference dropping
optimization. 
%

\vspace{-10pt}
\subsection{Baseline Evaluation}
\label{exp:baseline}

\begin{figure*}[t]
\vspace{-30pt}
	\centering
		\begin{subfigure}{\textwidth}
		    \centering
			\includegraphics[height=1.25in]{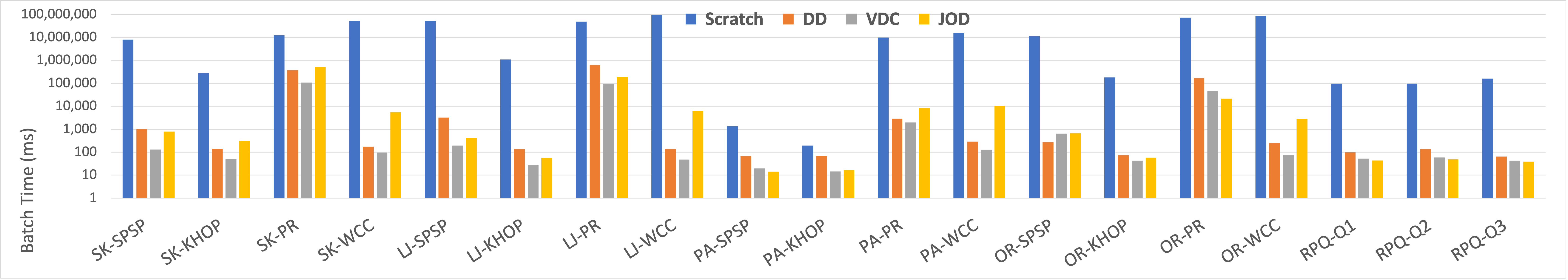}
			\caption{Total Batch Time (ms)}
		\end{subfigure}
	\newline
	\begin{subfigure}{\textwidth}
	    \centering
		\includegraphics[height=1.3in]{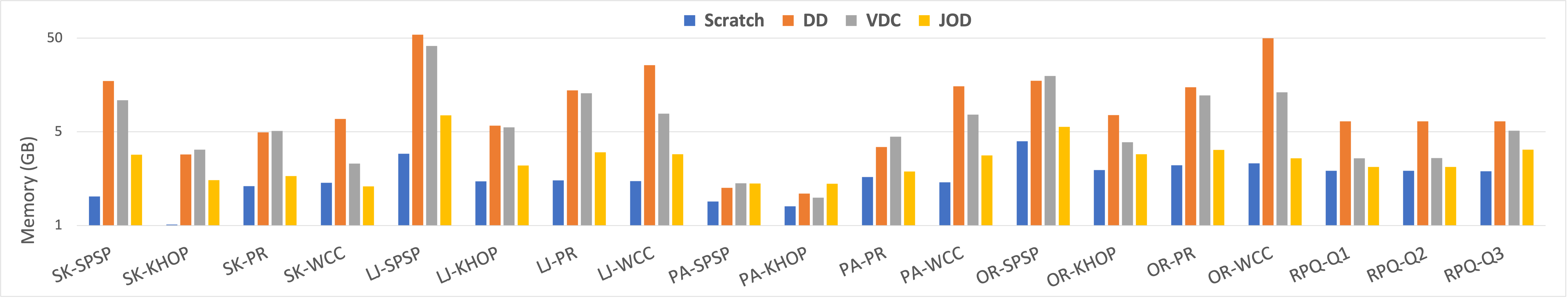}
		\caption{Memory (GB)}
		\label{fig:baseline-memory}
	\end{subfigure}
	\vspace{-10pt}
	\caption{Comparison between \SC, \DD, \VDC, and join-on-demand (\JOD).}
	\label{fig:Baselines}
	\vspace{-10pt}
\end{figure*}

Our first set of experiments measure the performances of
\SC, \DD, and \VDC. Our goals are: (i) to obtain baseline measurements for our
optimized \DC\ implementations; and (ii) 
to validate that \VDC\ is competitive with \DD\ to justify its use as a more suitable
baseline than \DD\ for our optimizations. 
In this experiment, we ran SPSP, K-hop queries, WCC and PR on Skitter, LiveJournal, Patents, and Orkut datasets, and all three RPQ queries on LDBC dataset. For SPSP, K-hop and RPQ workloads, we used 10 queries. In each experiment,
we simulated dynamism by using 100 insertion-only batches, with 1 edge in each batch. 

Our results are shown in Figure~\ref{fig:Baselines} (ignore the \JOD~charts for now). 
As shown in the figure \SC, as expected, is several orders of magnitude slower than \VDC\ and \DD\ but also has the
smallest memory overheads. 
\SC\ is most competitive with \VDC\ and \DD\ in PR, though still over an order of magnitude slower.
This is expected because as also observed in prior work~\cite{Graphsurge}, 
during differential maintenance, the changes in PR are harder 
to localize to small neighbourhoods as in other computations, i.e., small changes 
are more likely to change the PR values of larger number of vertices. 
We observe that \VDC\ is slightly faster than \DD\ while using comparable memory.
We expect \VDC\ to be faster than \DD\ because \DD\ assumes a distributed 
setting where messaging involves network protocols, even though we are running \DD\ in a single
machine setting. Instead, \VDC\ assumes a shared memory setting avoiding such 
communication.

\iftechreport
Appendix~\ref{exp:delete}
\fi
repeats these experiments with two different update workloads that include deletions: (i) where 25 of the batches are deletions; and (ii) where 50 of the batches are deletions. We observe that the performance tradeoffs our optimizations offer are broadly
similar across these different update workloads. Note that this is expected as the amount of updates we ingest is relatively minor compared to the number of edges we start with, which recall comprise 90\% of all edges in each dataset.
Overall these results confirm that \VDC\ is a more suitable baseline for analyzing the effects of our optimizations than \DD.
In the remainder, we use \VDC\ and \SC\ as the main baselines to evaluate our proposed optimizations on top of \VDC.

\begin{figure*}
	\centering
	\begin{subfigure}{0.66\columnwidth}
		\includegraphics[width=2in]{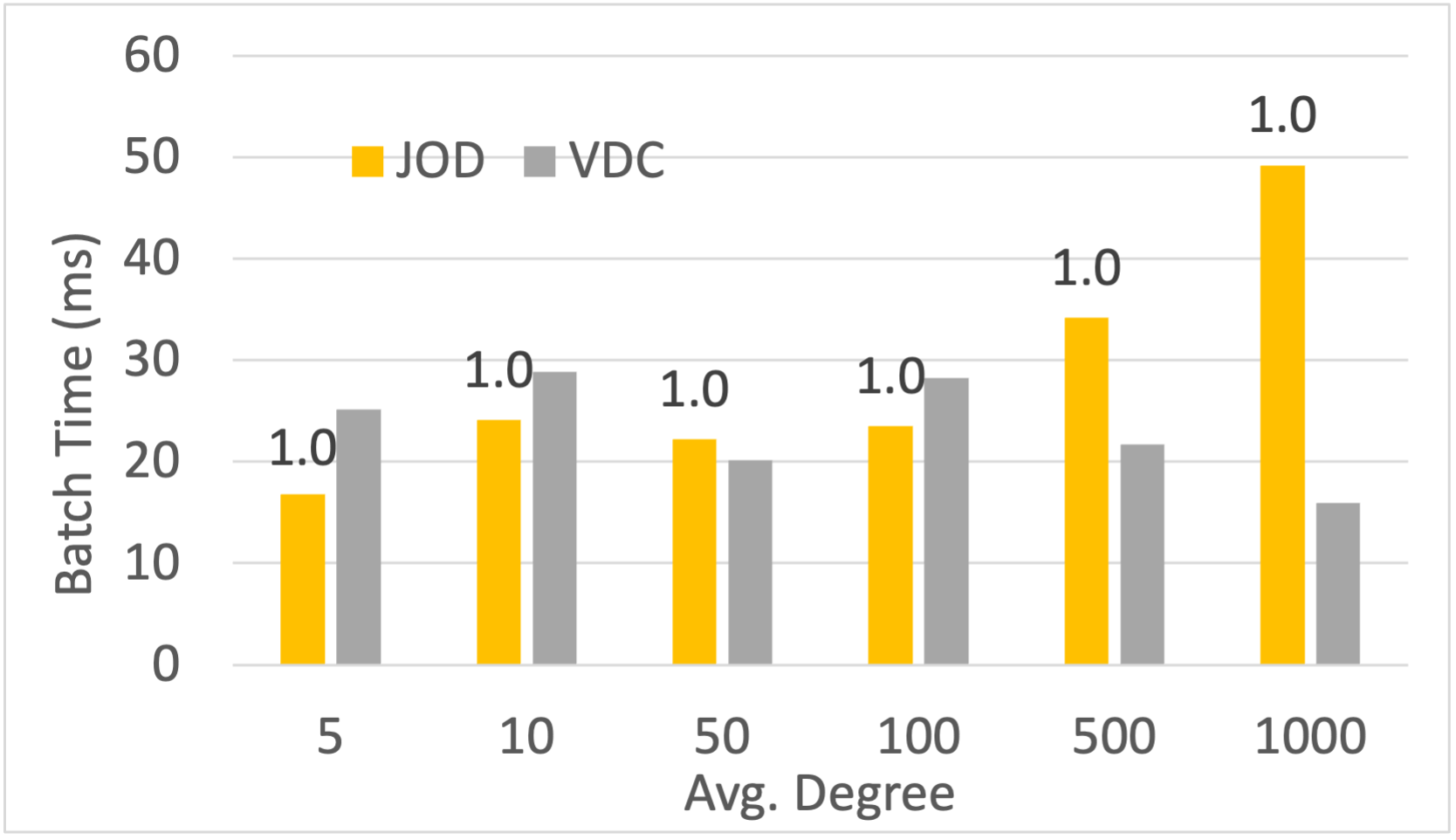}
		\caption{RPQ-Q1 on the \texttt{Knows} subgraph of LDBC.}
		\label{fig:Q1Knows}
	\end{subfigure}
	\begin{subfigure}{0.66\columnwidth}
		\includegraphics[width=2in]{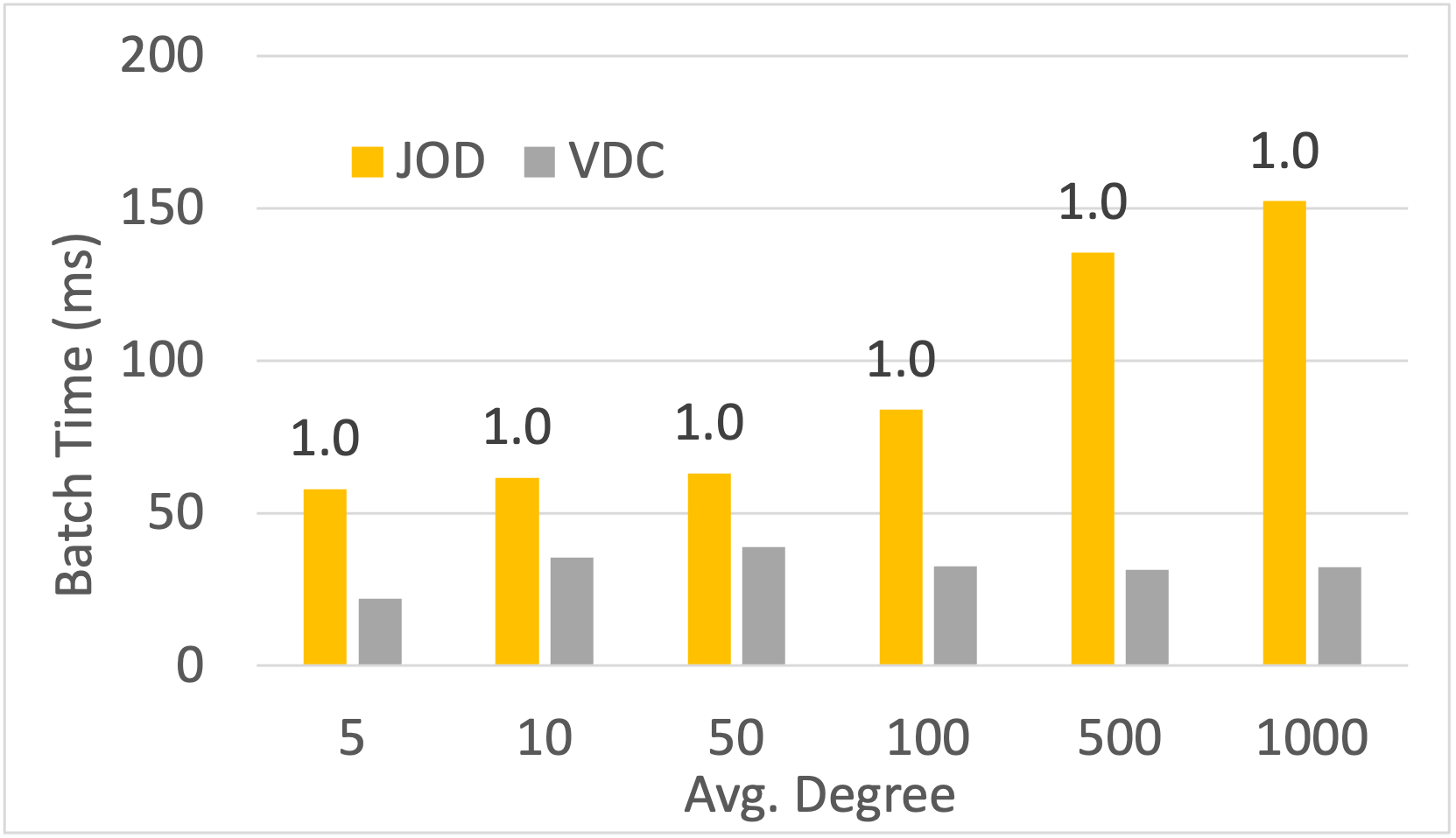}
		\caption{K-hop on the \texttt{Knows} subgraph of LDBC.}
		\label{fig:K-hopKnows}
	\end{subfigure}
	\begin{subfigure}{0.66\columnwidth}
		\includegraphics[width=2in]{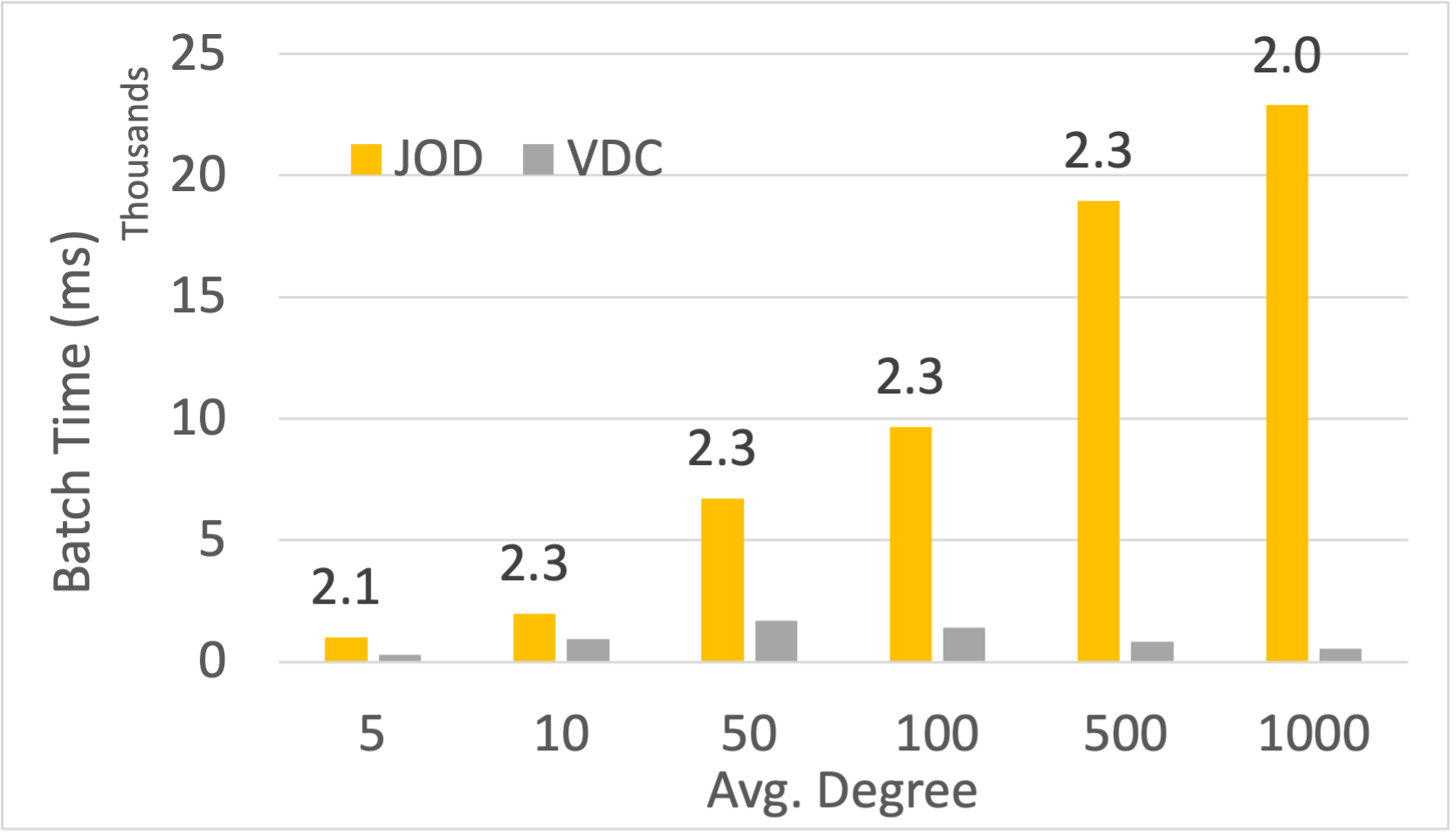}
		\caption{SPSP on the \texttt{Knows} subgraph of LDBC.}
	        \label{fig:LDBC_avgDegree}
	\end{subfigure}
	\vspace{-10pt}
	\caption{Comparison of \VDC\ and \JOD\ when running RPQ-Q1, K-hop, and SPSP as we increase the average vertex degree in the \texttt{Knows} subgraph of LDBC. Numbers on top of the are the average number of \diffs\ in $\delta D$ per vertex.}
	\vspace{-10pt}
	\label{fig:ldbc-degree-controlled}
\end{figure*}

\vspace{-10pt}
\subsection{Join-On-Demand}
\label{exp:CDD}
Our next set of experiments aim to study the performance and memory benefits and overheads of \JOD. \JOD~ is guaranteed to reduce the memory overhead of a system implementing vanilla differential computation, e.g. \DD\ or \VDC.
 However, in terms of performance, \JOD~ has both computation overheads and benefits. 
 On the one hand using \JOD\ reduces the work done by vanilla differential computation for storing differences. 
However, as updates arrive, \JOD\ requires re-computing the join on demand by 
reading the states of in-neighbours' of vertices at different timestamps to
inspect if some $\delta J$ partitions are non-empty.  This is less performant than materializing $\delta J$ difference sets
and inspecting them to see if they are non-empty. Our goal is to answer: \emph{What is the net effect of these performance benefits and costs? What governs this net effect?} 

Our hypothesis is that \JOD\ computation overhead increases proportionally with the average degree of the input graph. 
This is because, given a vertex $v$, looping through $v$'s incoming neighbours to re-compute the join
at a timestamp $t$ should increase with the number of neighbours of $v$.
At the same time the benefits 
of \JOD~ from not storing the differences depends on how many differences
are produced by the \texttt{Join} operator. This depends partially on average degree 
but also on the average number of times the state of a vertex changes during a computation. 
For example, readers can see that in the full difference trace of our running example, 
which
is presented in Table~\ref{fullTrace},
there is a new $\delta J$ difference
only when the state of a vertex changes.
As we will momentarily demonstrate, this number is quite small and does not necessarily
grow as the average degree increases on our computations. 
Therefore as the average degree increases, we expect that \JOD's overhead to increase faster than its benefits, 
and we should eventually see \VDC~outperforming \JOD~in terms of performance.

In our first experiment, we rerun our baseline experiments from Section~\ref{exp:baseline} with \JOD.
The average in-degrees Orkut, Skitter, LiveJournal, Patents, and LDBC (for the subgraph containing \texttt{Knows edges})
are respectively, $34.4$, $12.6$, $14.2$, $4.7$, and $4.7$. So
expect \VDC\ to be faster than \JOD\ by larger margins on Orkut and Skitter and smaller margins on Patents and LDBC.
Our results are shown in Figure~\ref{fig:Baselines}.
As expected, we observe that \JOD~ uses significantly less memory (between $1.2\times$ to $5.5\times$) than \VDC\ 
irrespective of the input graph or query. In terms of performance, we find as expected that \VDC\ is faster than \JOD\  on Orkut ($1.3$x on k-hop) and Skitter ($4.6\times$ on K-hop) and even slower than \JOD\ on Patents ($2.4$x on SPSP) and on LDBC RPQs (by a factor of $1.2\times$).

Although the previous experiment provides support for our hypothesis, the degree
differences between the input graphs we used are still relatively close to each other and we did not control for the 
queries we used across these datasets.
We next perform a more controlled experiment. Using LDBC, we systematically
increase the average degree of the \texttt{Knows} subgraph from its original value  ($4.7$ to $20$) to $100$, $500$, and $1000$ and run all of our algorithms SPSP, K-hop, and RPQ query Q1 on each version of these graphs. We increase the average degree by adding random edges that connect vertices in this subgraph. 

Our results are shown in Figure~\ref{fig:ldbc-degree-controlled}. 
As we expect, when the average degrees are small, specifically for RPQ queries, \JOD~either outperforms or is competitive with \VDC, but as the degrees get large, \VDC~consistently outperforms \JOD. 
 The numbers on top of the \VDC\ and \JOD\ bars in Figure~\ref{fig:ldbc-degree-controlled} are the average number of differences in $\delta D$ per vertex for vertices that have non-zero differences, measured at the end of the experiment. Note that this number is always 1.0 for K-hop and Q1 as vertex values in these computations take only one value assigned at the first iteration in which a vertex becomes reachable.
 
 \iftechreport
We note that this number does not even necessarily increase as the degree increases, and remains small relative to the average degree.  
It can even decrease in SPSP, primarily because SPSP converges faster when the degrees are larger, i.e., the number of SPSP iterations decrease, so the number of different  differences vertices get can decrease.
\fi

\begin{figure}
	\centering
	\begin{subfigure}{1\columnwidth}
		\centering
		\includegraphics[width=2.5in]{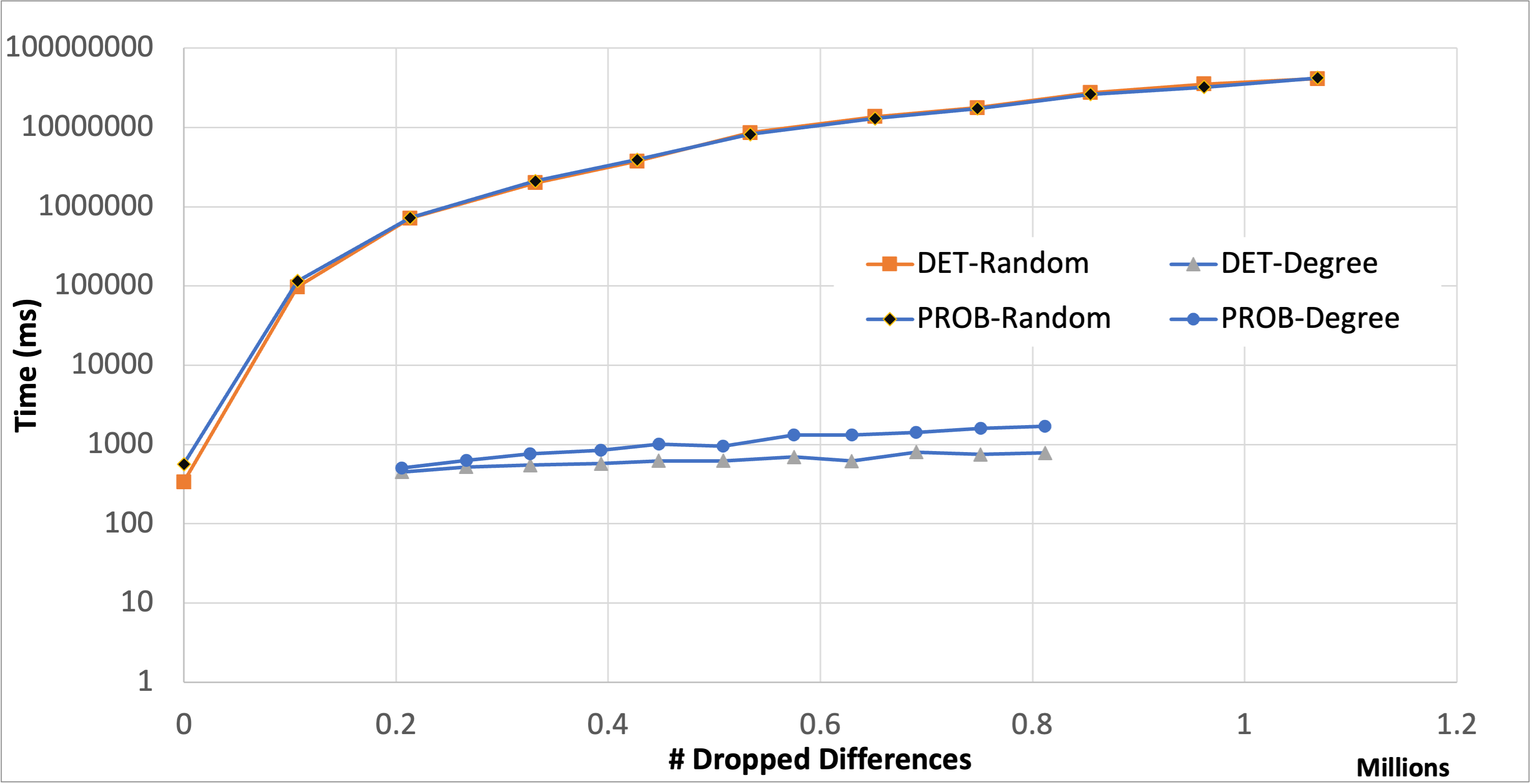}
		\caption{Performance of difference selection policies.}
		\label{fig:random-vs-selective}
	\end{subfigure}
\newline
	\begin{subfigure}{1\columnwidth}
		\centering
		\includegraphics[width=2.5in]{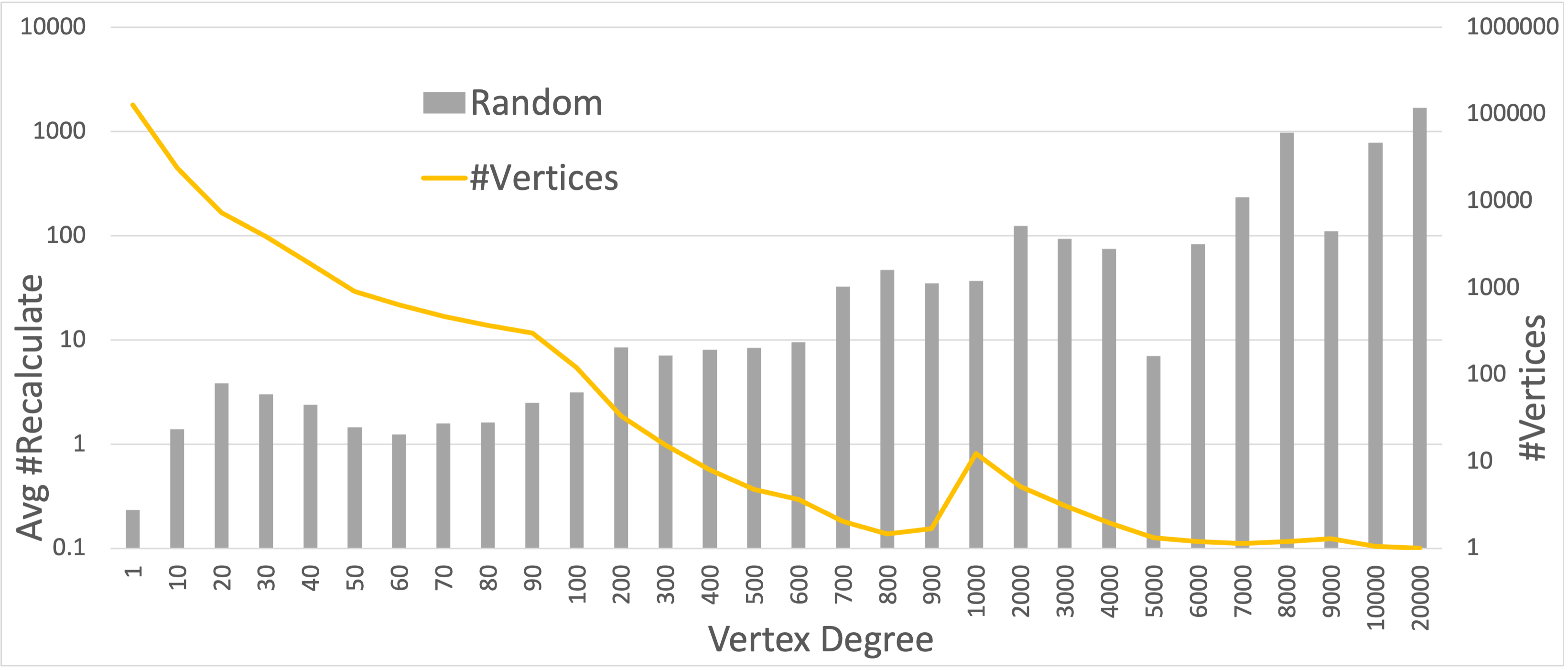}
		\caption{Average number of dropped differences re-computed per vertex.}
		\label{fig:degreeRecalculate}
	\end{subfigure}
	\vspace{-7pt}
	\caption{Comparison of \Random\ and \Degree\-based difference dropping when running 10 K-hop queries.}
	\label{fig:howToDrop}
	\vspace{-20pt}
\end{figure}

\vspace{-10pt}
\subsection{Selecting the Differences To Drop}
\label{exp:Random}

We next evaluate the effectiveness of the two strategies we discussed in Section~\ref{sec:selectingDiffs} for selecting
which differences to drop in our partial dropping optimization. We refer to these as: (i) \Random\, which 
randomly selects the differences with a probability $p$; and (ii)  \Degree\ drops differences based on
vertex degrees. 
As we discussed in Section~\ref{sec:degree-drop}, we expect \Degree\ to outperform \Random.

We run 10 K-hop queries over Skitter with 100 insertion-only batches of size 1 using \DET\ and \PROB\ with both \Random\ and \Degree\  selection strategies. In total, we have 4 system configurations. For \Degree, we set $\tau_{min}$ to 2 and
$\tau_{max}$ to the $80^{th}$ percentile of the vertex degrees. 
We increase the dropping probability $p$ for \DET\ and \PROB\ starting
from 0 to 100\% and plot the total number of dropped differences on the $x$-axis and the runtime on the $y$-axis. 
Figure~\ref{fig:howToDrop} shows our results.
First, observe that as expected all of the lines in the figure go up, i.e.,
as we drop more differences the performance of each system configuration gets slower. Note that in \JOD\ storing fewer differences potentially leads to performance advantage as we have to maintain less differences. This advantage does not exist for partial dropping optimizations because they still have to store and maintain auxiliary data structures to maintain the dropped differences. 
So dropping differences primarily has a performance cost, as it can lead the system to recompute those dropped differences.
Second observe that as we expect, configurations with 
\Degree\ (the two bottom lines), irrespective of if we use 
 \DET\ and \PROB, are 
between $3$ to $5$ orders of magnitudes faster than the configurations with \Random\ (two top lines).
Note that the lines with \Random\ have a bigger span on $x$-axis 
because there are limits to the minimum and maximum number
of differences that configurations with \Degree\ can drop.
For example, at the minimum when $p=0$, the configurations with
\Degree\ still drops all differences of vertices with  degree $< \tau_{min}$, 
whereas \Random\ can drop as few as 0 differences.

We  perform further analyses using a micro-benchmark to better explain the performance difference 
between \Random\ and \Degree. We first fix the drop probability $p$ ($=0.1$), a workload (10 K-hop queries) and a dataset (Skitter with 100 batch of 1 edge insertions).
We then use \DET\ with \Random\ selection policy and count 
for each vertex $v$ the number of times \DET\ re-computed a dropped difference with key $v$, 
i.e., how many times \DET\ has accessed $D^v$ at some point, but $v$'s state had to be re-computed
because a difference was dropped
in \texttt{DroppedVT}. Then we bucket vertices by their degree,
where for each degree bucket (e.g., $[1-10)$) we plot the average number of re-computations for each vertex
in that bucket. 
Figure~\ref{fig:degreeRecalculate} shows our results.
The bar charts use the left $y$-axes and 
represent the average number of re-computations for vertices with different degree buckets, where a tick in the $x$-axes 
represents a bucket with the next tick. 
The line chart uses
the right $y$-axes and plots the vertex degree distribution in the graph.

As shown in Figure~\ref{fig:degreeRecalculate}, the degree
distribution follows a power-law distribution,  as is commonly the case in real world graphs. 
The average number of re-computations per vertex follows the opposite trend where vertices
with smaller degrees on average lead to fewer re-computations, e.g, vertices with degree more than $2000$ lead
to more than $1000$ re-computations on average, while those with degrees $[1, 10)$ lead to less than $1$ re-computations.
Since the memory saving of dropping $1$ \diff\ is the same regardless of the vertex degree, 
as done by our \Degree\ strategy,
it is more efficient
to drop more \diffs\ from vertices with smaller degrees. 

\begin{figure*}
	\centering
	\begin{subfigure}{0.66\columnwidth}
		\includegraphics[width=2.2in]{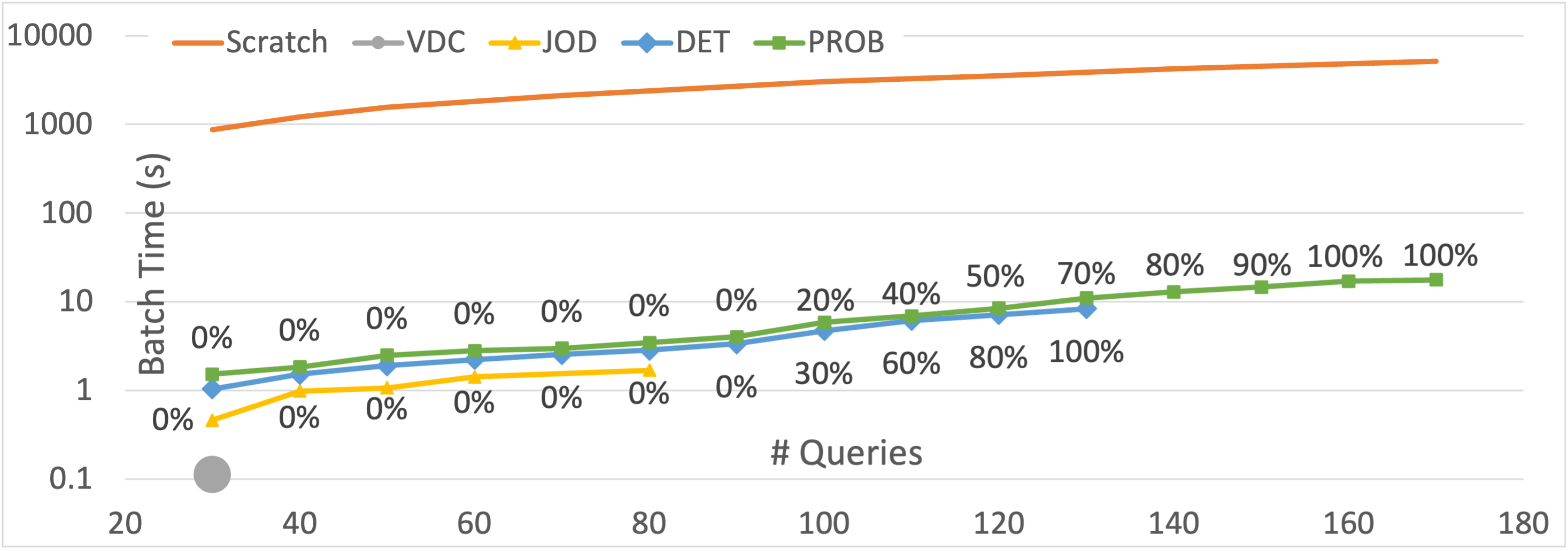}
		\caption{K-hop query in SK dataset}
		\label{fig:SK-K-hop}
	\end{subfigure}
	\begin{subfigure}{0.66\columnwidth}
		\includegraphics[width=2.2in]{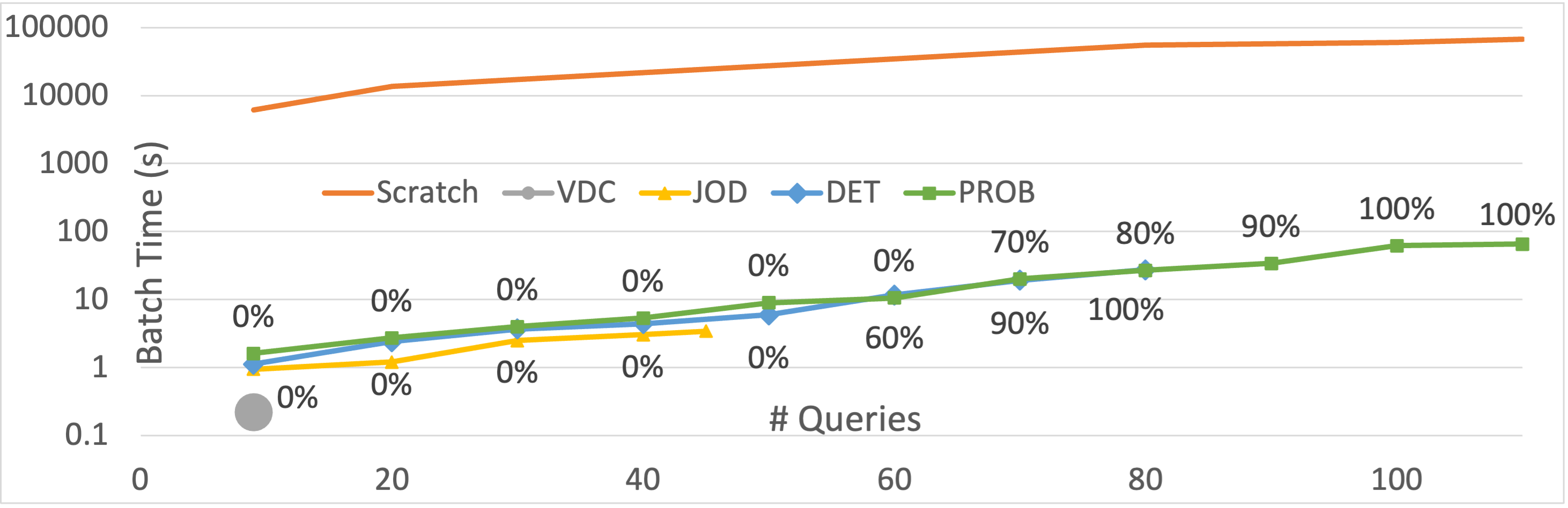}
		\caption{SPSP query in SK dataset}
	\end{subfigure}
	\begin{subfigure}{0.66\columnwidth}
		\includegraphics[width=2.2in]{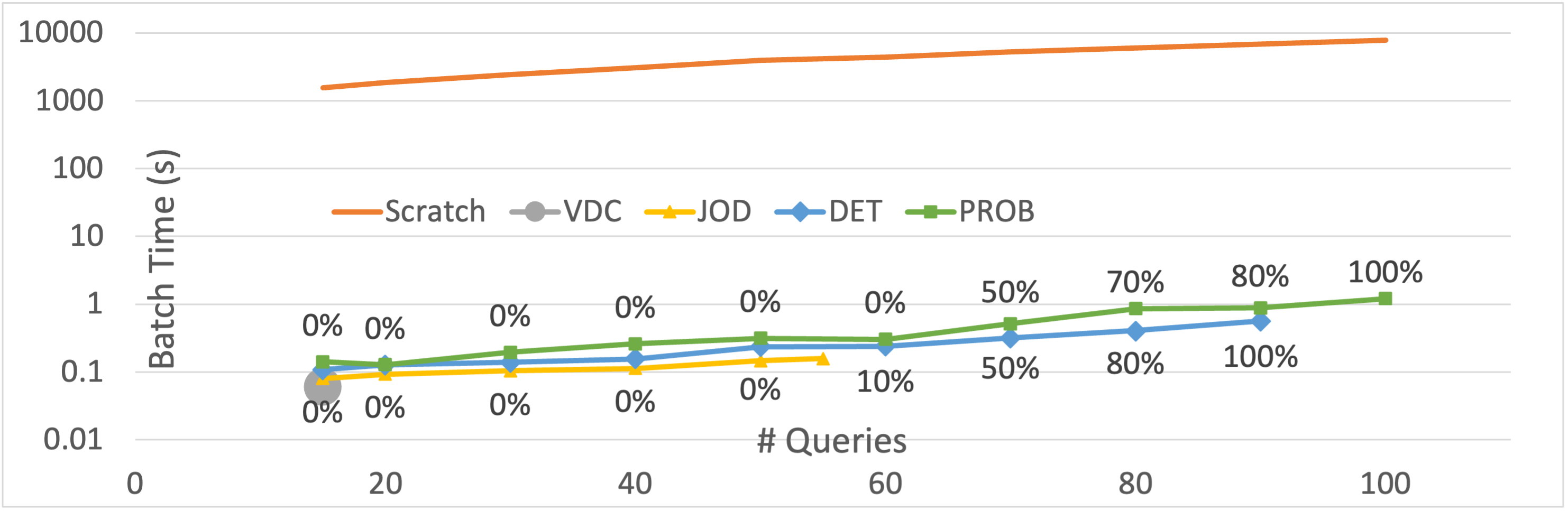}
		\caption{K-hop query in LJ dataset}
	\end{subfigure}
	\begin{subfigure}{0.66\columnwidth}
		\includegraphics[width=2.2in]{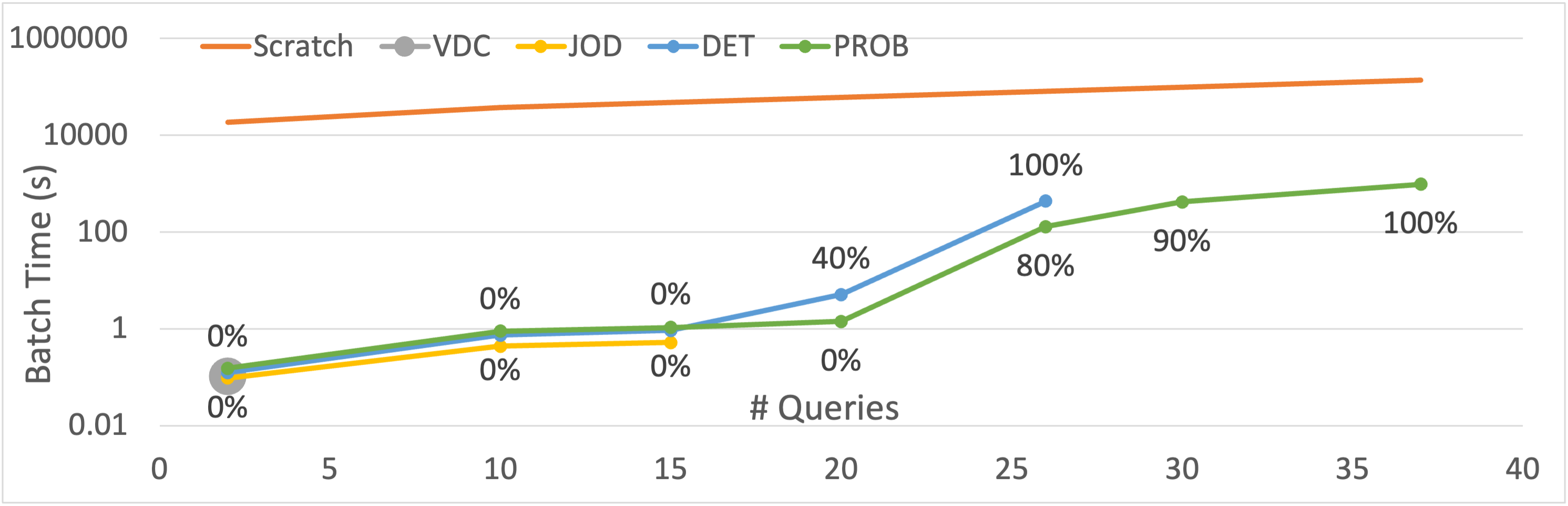}
		\caption{SPSP query in LJ dataset}
		\label{fig:LJ_SPSP}
	\end{subfigure}
	\begin{subfigure}{0.66\columnwidth}
		\includegraphics[width=2.2in]{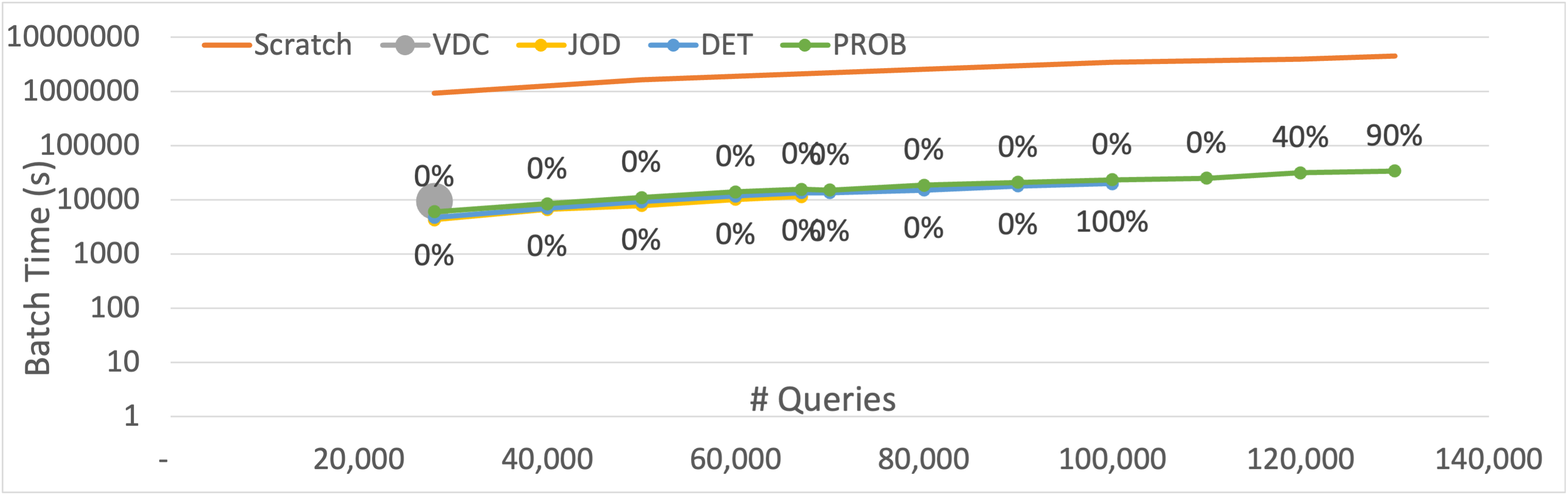}
		\caption{K-hop query in PA dataset}
	\end{subfigure}
	\begin{subfigure}{0.66\columnwidth}
		\includegraphics[width=2.2in]{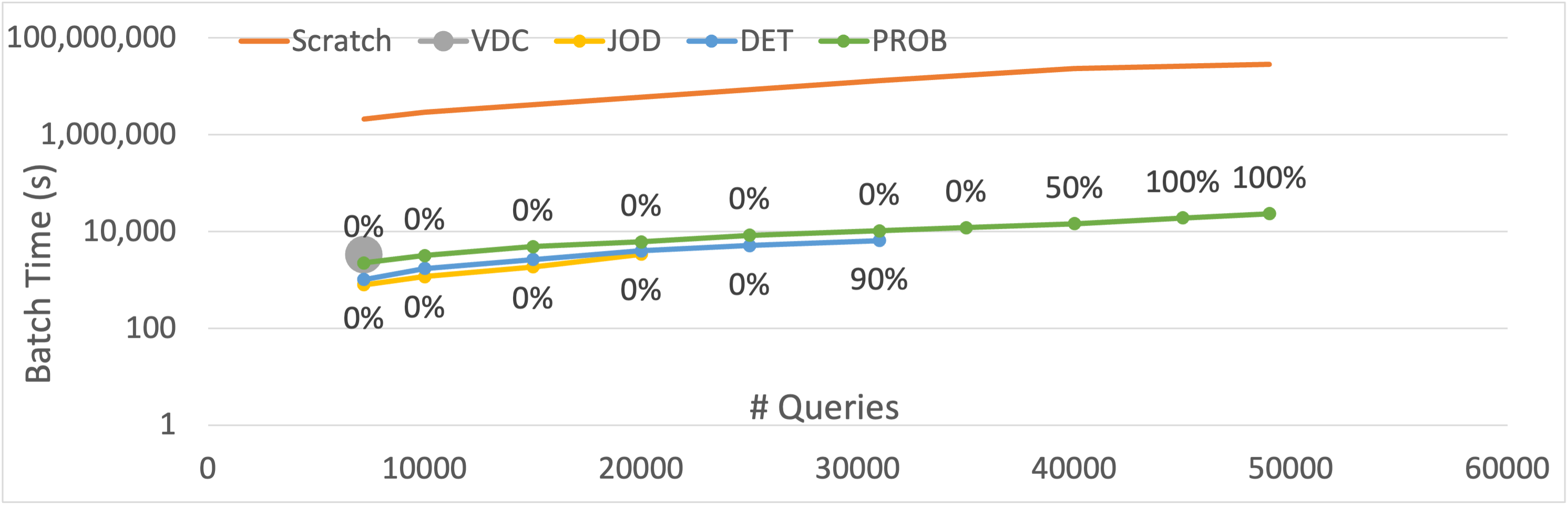}
		\caption{SPSP query in PA dataset}
		\label{fig:PA_SPSP}
	\end{subfigure}
	\begin{subfigure}{0.66\columnwidth}
		\includegraphics[width=2.2in]{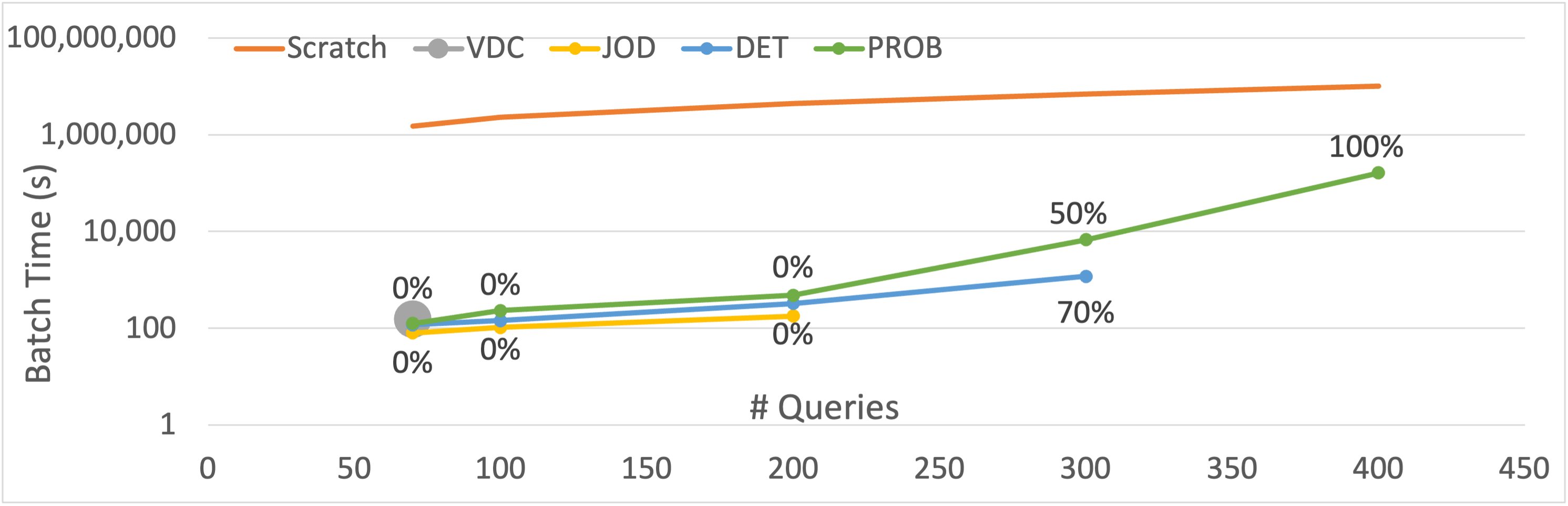}
		\caption{K-hop query in OR dataset}
	\end{subfigure}
	\begin{subfigure}{0.66\columnwidth}
		\includegraphics[width=2.2in]{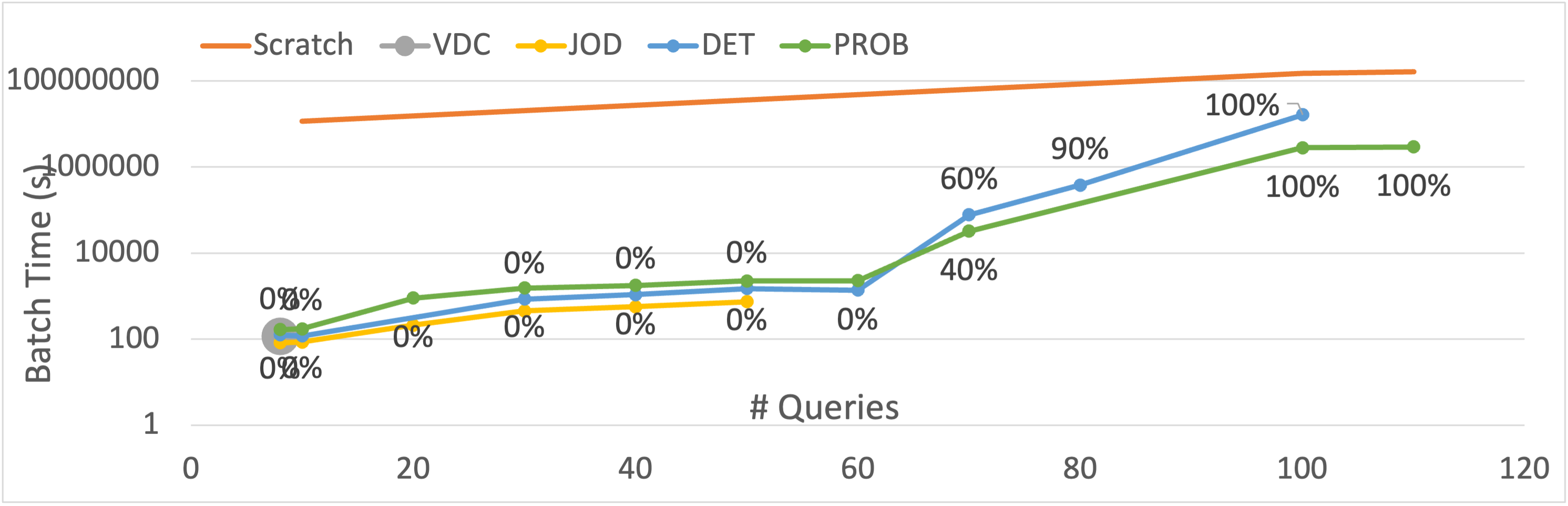}
		\caption{SPSP query in OR dataset}
		\label{fig:OR_SPSP}
	\end{subfigure}
	\begin{subfigure}{0.66\columnwidth}
		\includegraphics[width=2.2in]{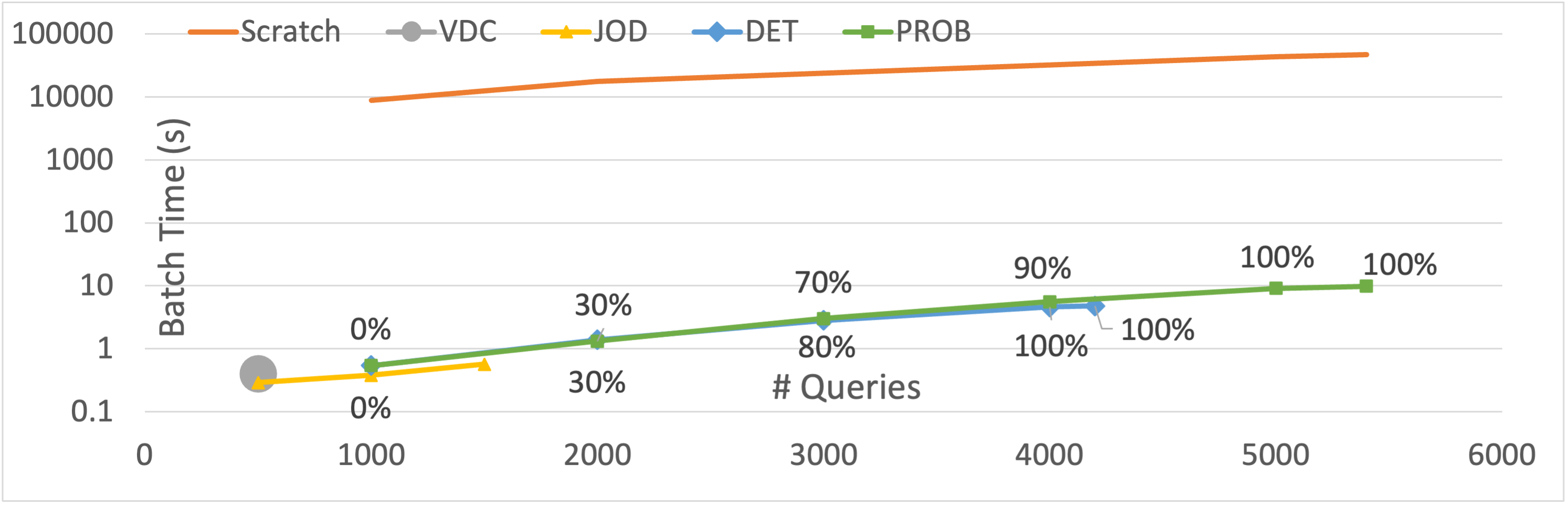}
		\caption{RPQ-Q1}
		\label{fig:Q1}
	\end{subfigure}
	\begin{subfigure}{0.66\columnwidth}
		\includegraphics[width=2.2in]{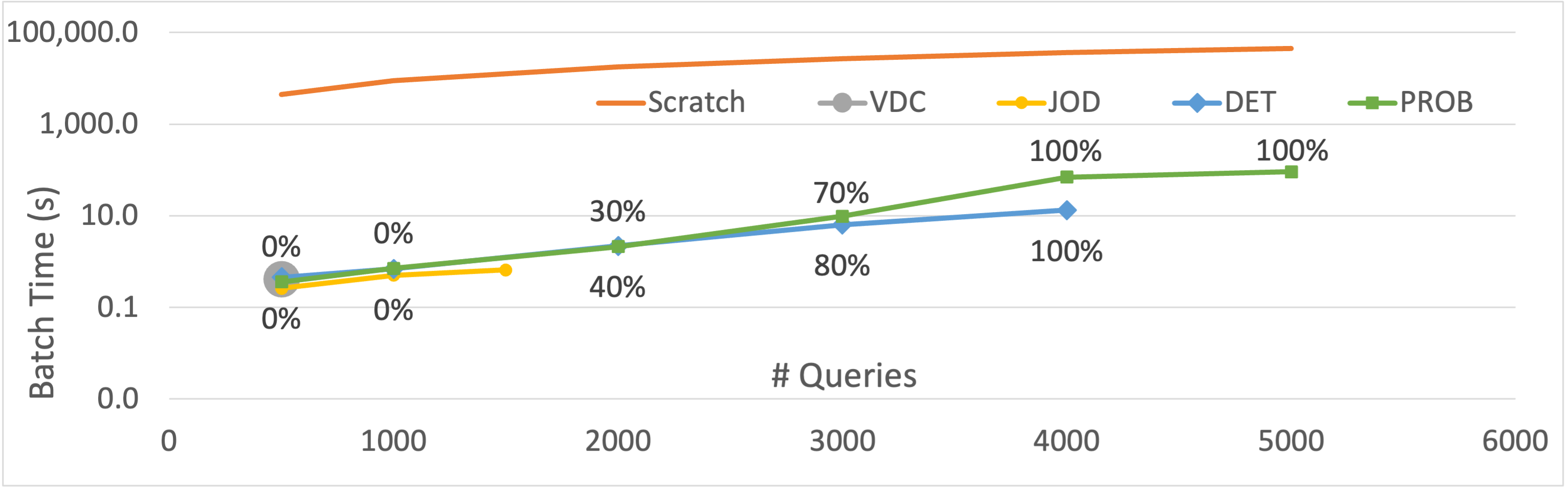}
		\caption{RPQ-Q2}
	\end{subfigure}
	\begin{subfigure}{0.66\columnwidth}
		\includegraphics[width=2.2in]{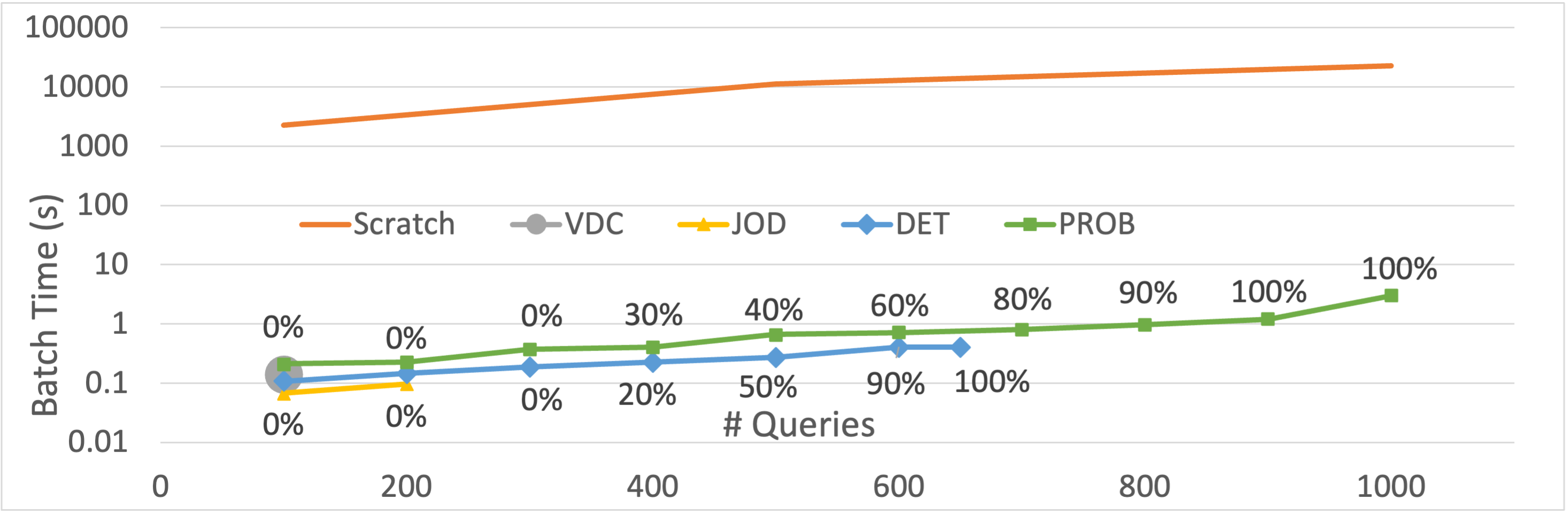}
		\caption{RPQ-Q3}
	\end{subfigure}
	\caption{Number of queries maintained by \SC, \DC, \JOD, \DET, and \PROB\ under a limited memory budget of 10GB. The large dot in bottom left of each figure is \DC.}
	\label{fig:scalability}
\end{figure*}

\vspace{-15pt}
\subsection{Difference Maintenance}
\label{exp: partialDrop}
\label{exp:scalability}
Our next set of experiments focus on evaluating \DET\  and \PROB. 
In the experiments reported in Figure~\ref{fig:random-vs-selective}, we evaluate the performances of \DET\ and \PROB\ when both drop exactly the same number of differences when using \Degree\ and \Random\ selection policies. 
They behave similarly when using the same selection strategy, with \DET\ slightly more performant, which is expected as \PROB\
may perform spurious re-computations due to false positives. However, 
\DET\ and \PROB\ do not have similar memory footprints when they drop
the same number of differences: \PROB's approach 
is more efficient than \DET. We next provide a more systematic evaluation of the scalability and performance tradeoffs of
these techniques under \Degree\ policy, which as we established outperforms \Random.

Our experiment analyzes how much \DET\ and \PROB\ increases the system scalability in terms of the number of concurrently maintained queries relative to \VDC\ for a given memory budget
 for SSSP, K-hop, and RPQ queries. We omit PageRank and WCC from these experiments, as these are batch computations and we cannot increase the number of queries for these.
For completeness, we also evaluate the performances of \JOD\
and \SC.
To simulate a fixed memory budget environment, we give each system configuration $10$GB memory for storing differences
and/or additional data structures, e.g., to manage dropped VT pairs. 
We repeat our experiment from Section~\ref{exp:baseline} with the same datasets and query combinations. 
However, we now increase the number of queries systematically until the system runs out of memory.

Figure~\ref{fig:scalability}  shows our results. 
We use the maximum scalability level of \VDC, which is the configuration with the highest memory overheads, 
as the lowest number of queries we use and increase the number of queries in the system from this point on. That is why
\VDC\ appears as a single grey point in our charts.
 For \DET\ and \PROB, for each number of queries $q$, 
we find the lowest dropping probability $p_{det}$ for \DET\ and $p_{prob}$ for \PROB\ 
that can support $q$ queries and
report their performances with these levels. 
Note that here we are assuming an ideal setting in which a system 
is able to find this lowest dropping probability. Although this may be challenging
in practice, this allows us to evalaute the most 
performant versions of \DET\  and \PROB\ for the given query level.
We show $p_{det}$ that we use for \DET\ under the \DET\ line, and the $p_{prob}$ that is used
for \PROB\ above the \PROB\ line.

We make several observations. First, as in Figure~\ref{fig:Baselines},  we see that
\JOD\ can increase the number of queries that could be concurrently run by $2.3$$\times-10 \times$ over \VDC. 
Second, we observe that increasing the number of queries with partial dropping optimizations 
can increase the run time super-linearly beyond a particular point where  increasing scalability requires increasing the dropping probability, which leads to more \diffs\ to be re-computed.
However, we see that partially dropping differences can still increase the number of concurrent queries by up to $20\times$ relative to \VDC\ while still outperforming \SC\ by several orders of magnitude. 
Third, we compare the performances of \DET\ and \PROB.
As mentioned earlier, \DET\ does not incur any spurious re-computations due to false positives but 
has to drop more differences than \PROB\ to scale to more queries (as it has a higher memory
overhead for storing the dropped VT pairs). We see that this 
advantage and disadvantage overall balance out 
for the scalability levels both \DET\ and \PROB\ can handle, i.e., they perform similarly at these scalability levels.
However, \PROB\ can consistently scale to higher levels than \DET\ (up to  $1.5 \times$).

Finally, we performed a similar experiment for PR and WCC, for which we can only run one ``query''.
We used LJ and picked a memory budget of 2.75GB for PR and 2GB for WCC, which requires less memory
and picked the lowest drop probabilities at which these budgets were enough for \DET\ and \PROB. 
Figure~\ref{fig:detvsProb} shows our results, with the necessary drop percentages presented
on top of the bars. 
We find that on PR \DET\ requires 100\% dropping rate and takes
$369$ seconds to complete
 while \PROB\ requires 90\% dropping rate and takes $268$ seconds to complete\footnote{Recall
that 100\% dropping rate does not mean all differences are dropped as we do not drop any differences 
for vertices over 13 degree in LiveJournal dataset.}. 
On WCC, \DET\ requires 90\% dropping rate and takes
$11.9$ seconds to complete
 while \PROB\ requires 70\% dropping rate and takes $11.5$ seconds to complete.
Overall, similar to our previous experiments, \PROB\ needs to drop fewer 
differences to successfully complete the experiment and leads to better performance.

\begin{figure}[tb]
	\includegraphics[width=1.5in]{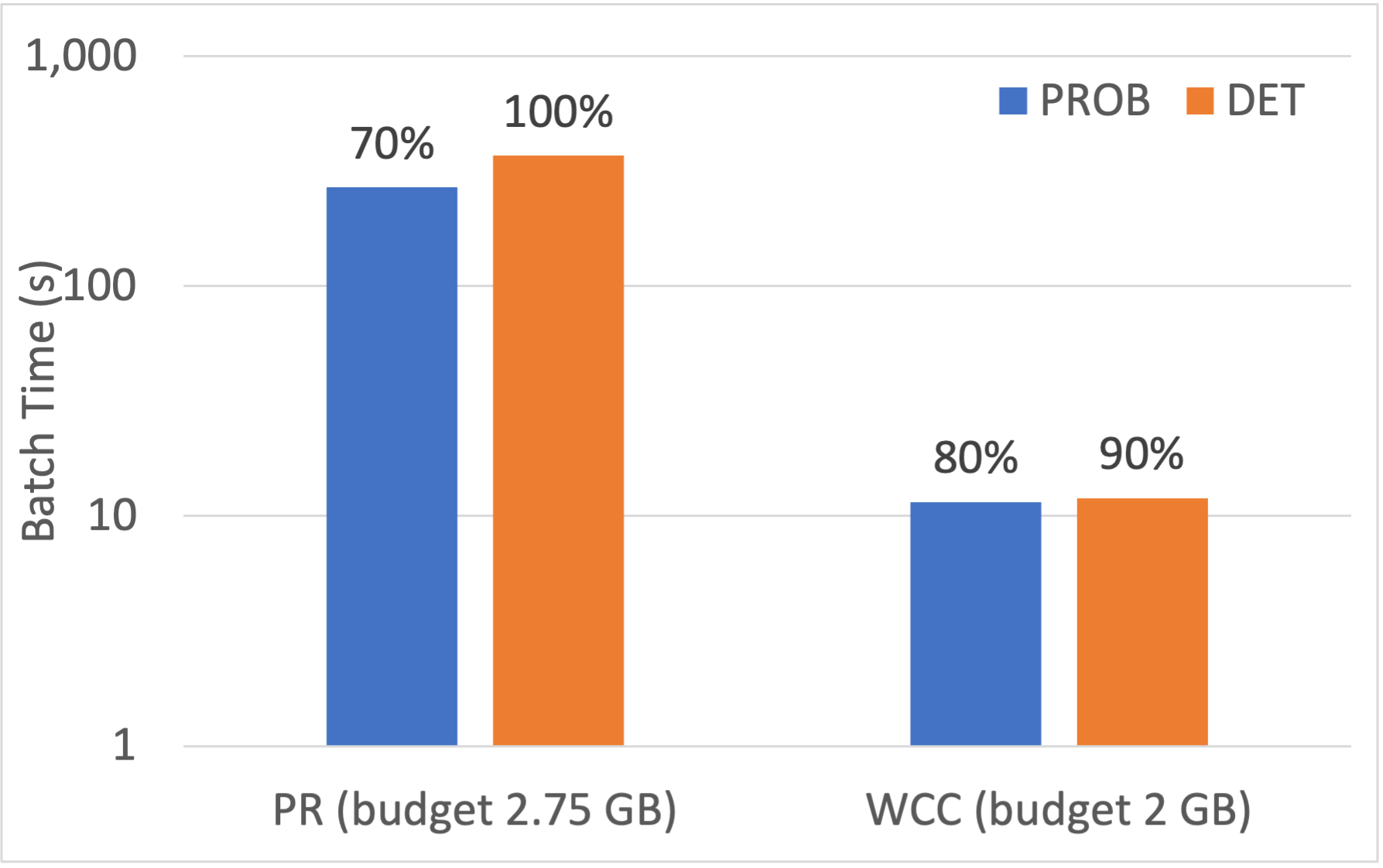}
	\caption{Comparison of \DET\ and \PROB\ when running PageRank and WCC on LJ under limited memory.}
	\label{fig:detvsProb}
	\vspace{-15pt}
\end{figure}

%

\subsection{Further Applications of Diff-IFE}
\label{sec:landmark}

Our previous experiments so far focused on demonstrating the performance tradeoffs that our optimizations offer
when evaluating continuous recursive queries using Diff-IFE. Our final set of experiments do not evaluate our optimizations.
Instead, we aim to demonstrate further applications of Diff-IFE in systems.
Specifically, we show that we can improve our \SC\ baseline for SPSP queries through using and differentially maintaining
a popular shortest path index, called {\em landmark indices}~\cite{graphSearchA,landmark-approximate-CIKM}. A landmark index is a single-source shortest distance index, i.e.,
it stores the shortest path distance from a ``landmark'' vertex to the rest of the vertices. 
We use landmark indices to prune the search space of \SC. 
Specifically, in the 
shortest path query from $s$ to $d$, the sum of the distances of $s$ to $l$ and $d$ to $l$ give an upper
bound $\ell_u$ on the shortest distance between $s$ and $d$. Similarly, the difference between the $v$ to $l$ distance 
and $d$ to $l$ distance give
a lower bound $\ell_b$ on the distance from $u$ to $d$. If $v$ is visited at distance $k$ in the Bellman-Ford algorithm,
and $k +\ell_b$ is greater than $\ell_u$, then we can avoid traversing $v$ as it cannot be on the shortest path from 
$s$ to $d$. 

We used all of our datasets, except LDBC, and picked the 10 highest-degree nodes as the landmarks and implemented an optimized
version of \SC\ in which as updates arrive at the graph, we first maintain these 10 landmark indices using Diff-IFE. 
Then, we run each registered query using our landmark-enhanced \SC, which we call \SC-landmark,
and compare this to our baseline \SC. 
We registered 100 random SPSP queries in our system and measured the end-to-end time of 100
batches of single edge insertions. Our results are shown in
Figure~\ref{fig:landmark}. The reported times for \SC-landmark include both the time to maintain the index and then (non-differentially)
evaluate each query. As shown in the figure, by using and differentially maintaining landmark indices, we can reduce \SC\ time between $43\%$ to $83\%$ (albeit now using additional memory to store both the index and the differences to differentially 
maintain the index). 


\begin{figure}
	\includegraphics[width=2in]{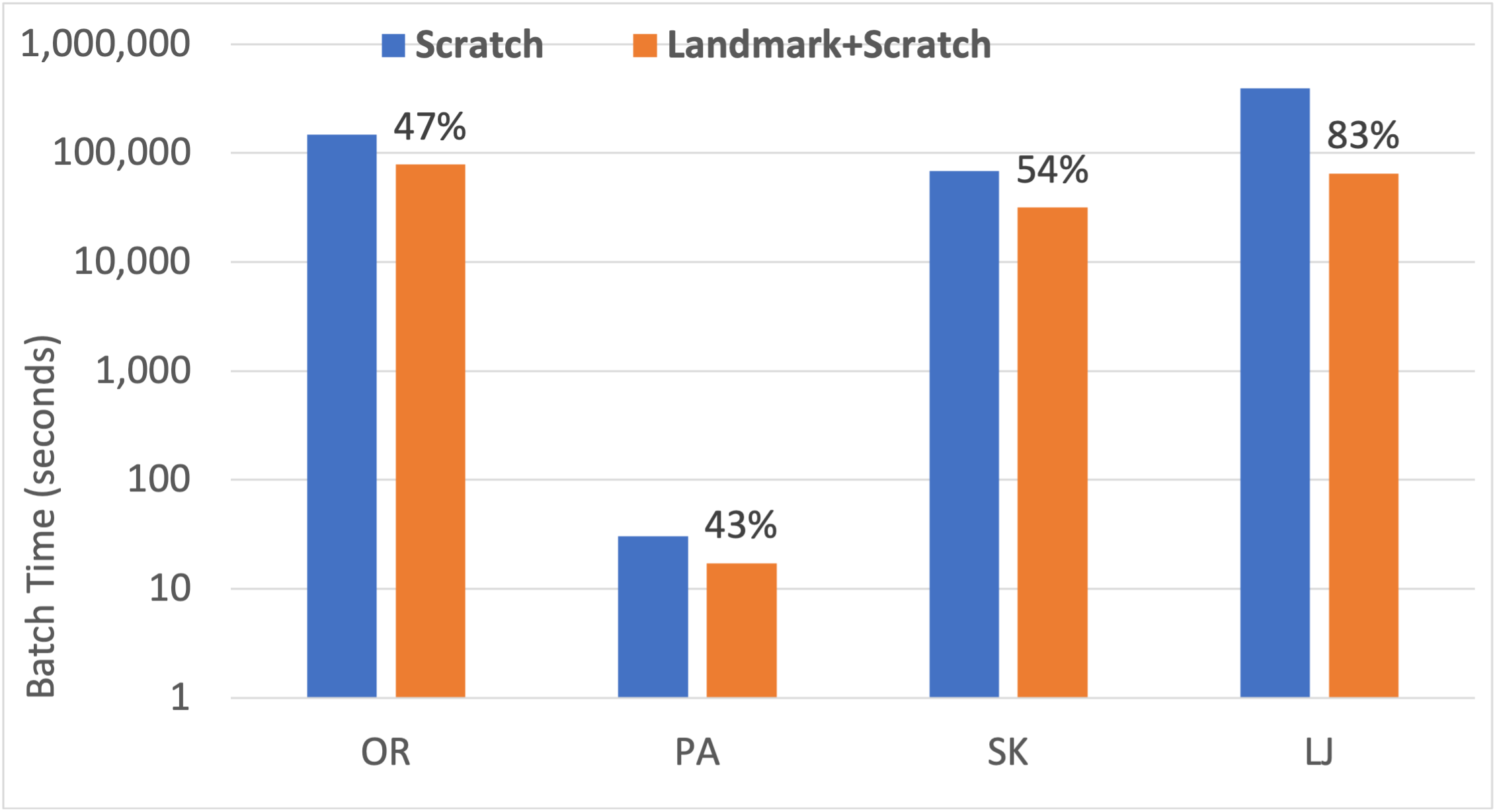}
	\vspace{-10pt}
	\caption{Comparing \SC\ vs. \SC-landmark on 100 queries and 100 batches of updates.  Numbers on orange bars are the runtime improvements of \SC-landmark. }
	\label{fig:landmark}
	\vspace{-25pt}
\end{figure}

%% file: future-work.tex
\section{Conclusions}
\label{sec:conclusion}

Differential computation is a generic novel technique to maintain arbitrary recursive dataflow computations. 
As such, it is a promising technique 
to integrate into data management systems that aim to support continuous recursive queries.
We studied the problem of how to increase scalability of differential computation 
through optimizations that are based on dropping differences.

An important future work venue is to design more advanced algorithms than vanilla IFE that can be maintained 
efficiently with differential computation.
One example is to differentially maintain shortest path algorithms that use indices.
In Section~\ref{sec:landmark}, we demonstrated how Diff-IFE can be used to maintain landmark indices, 
but this algorithm only enhanced our baseline \SC\ algorithm with an index and does not evaluate queries
differentially. It is less clear how to design a differential shortest path algorithm that uses an index that also needs
to be updated as updates arrive at the system. One possibility is to develop two separate dataflows:
(1) that maintains the indices; (2) that uses the updates to E and the updates to indices as base collections, 
and uses both the index
and edges, e.g, by joining E, distances D, and the indices, to find shortest paths. 
No prior work we are aware of has proposed such algorithms and developing them is an important research topic.

%% file: Appendix.tex
\appendix

\section{Impact of Batch Size}
\label{sec:batchSizeImpact}

\begin{figure}[b]
	\centering
	\includegraphics[height=1.5in]{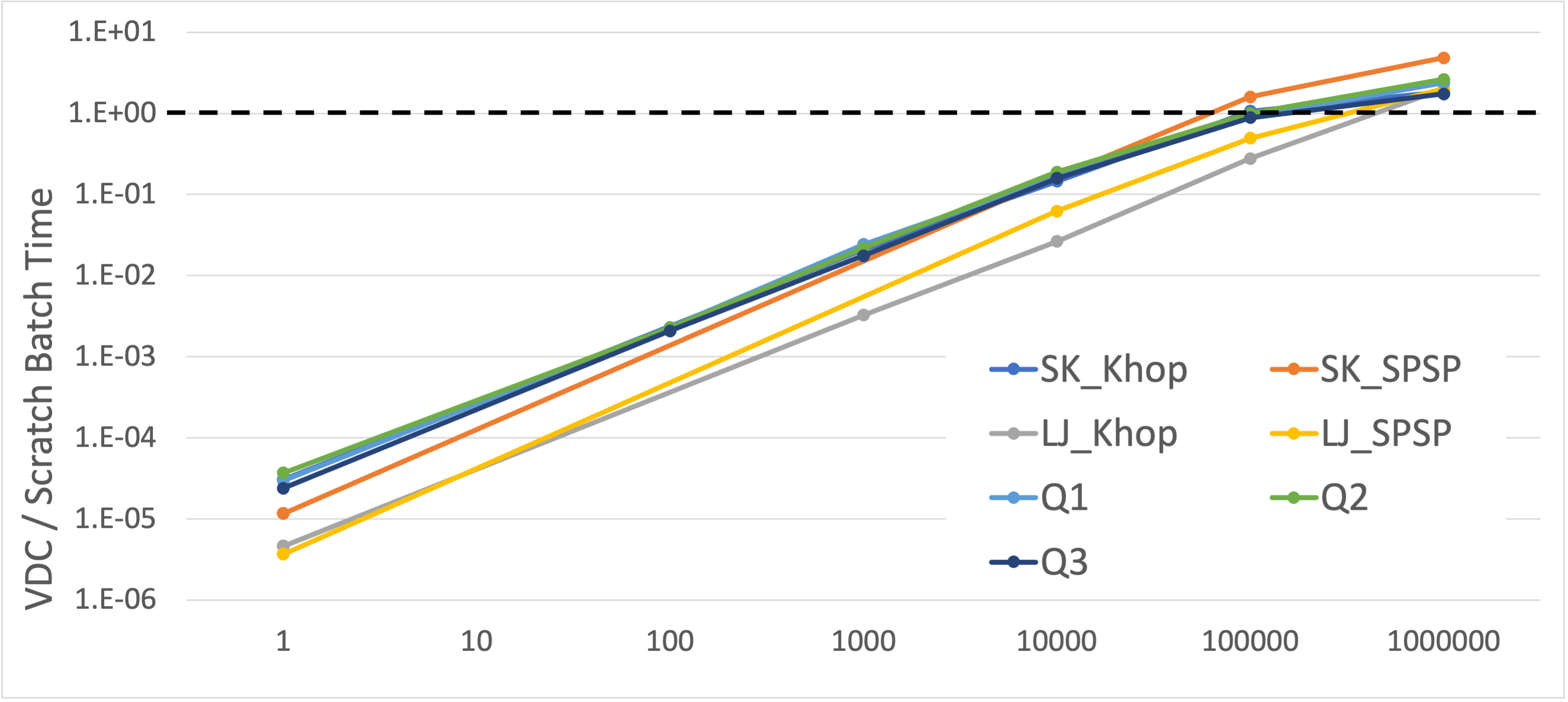}
	\caption{Influence of batch size on \DC\ performance. The $x$-axis is batch size, and $y$-axis is the ratio between \VDC\ and \SC\ batch time. A small batch size can lead \VDC\ to be several order of magnitude faster than \SC. As the batch size gets larger, \VDC's cost increases until the ratio passes 1 and \SC\ becomes faster than \VDC.}
	\label{fig:batchSizeImpactOnExecution}
\end{figure}

%

This set of experiments show the impact of batch size in the performance of \DC. We start by load the initial graph and register 10 K-hop queries, then add $1$M edge updates while changing the batch size exponentially from $1$, $10$, $100$, to $1$M. Figure~\ref{fig:batchSizeImpactOnExecution} shows the ratio between the execution time of each batch by \VDC\ and by \SC. When this ratio is more than 1, it means that \VDC\ is slower than \SC. This happens when the batch size is very large (more than 100K in our experiments).



Changing the batch size has no significant impact on \SC\ performance, because it re-executes the query from scratch anyway. In Figure~\ref{fig:batchSizeImpactOnExecution} \VDC\ gets slower as the batch size increases due to processing larger number of \diffs\ in every batch. Increasing the batch size means that the effort required by \VDC\ to maintain all differences increases, because the number of potential vertices that need to be fixed increases. It is apparent that the smaller is the batch size the better is \VDC\ performance, which shows that \DC\ is more suitable for near real-time dynamic graph updates than to infrequent updates. 


\section{Impact of delete batches}
\label{exp:delete}

All previous experiments has been assuming edge addition. In this section, we are evaluating the impact of delete batches. Figure~\ref{fig:baseline-delete} shows the baseline experiments when $25\%$ and $50\%$ of batches are deleting edges. These figures are very similar to the baseline figure (Figure~\ref{fig:Baselines}) where all batches are edge additions. 

\begin{figure*}[t]
	\centering
		\begin{subfigure}{\textwidth}
		    \centering
			\includegraphics[height=1.25in]{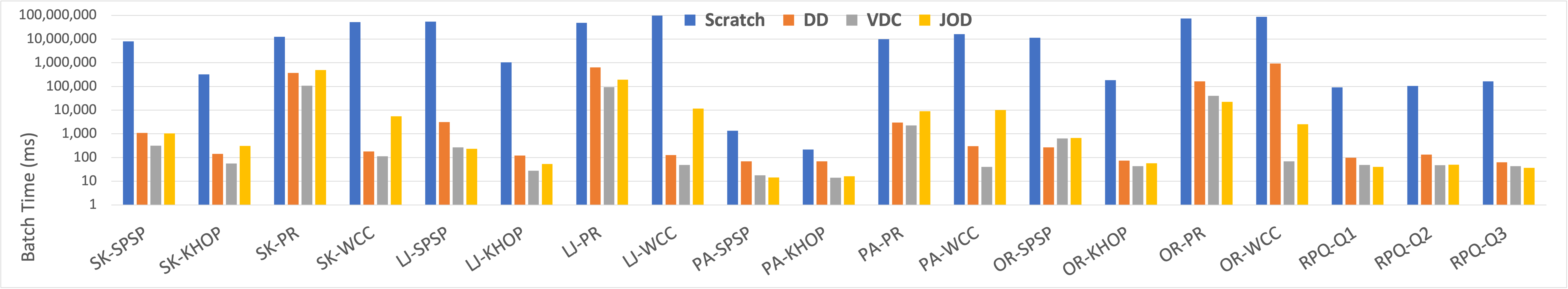}
			\caption{Total Batch Time (ms) with $25\%$ batches}
		\end{subfigure}
	\newline
	\begin{subfigure}{\textwidth}
	    \centering
		\includegraphics[height=1.25in]{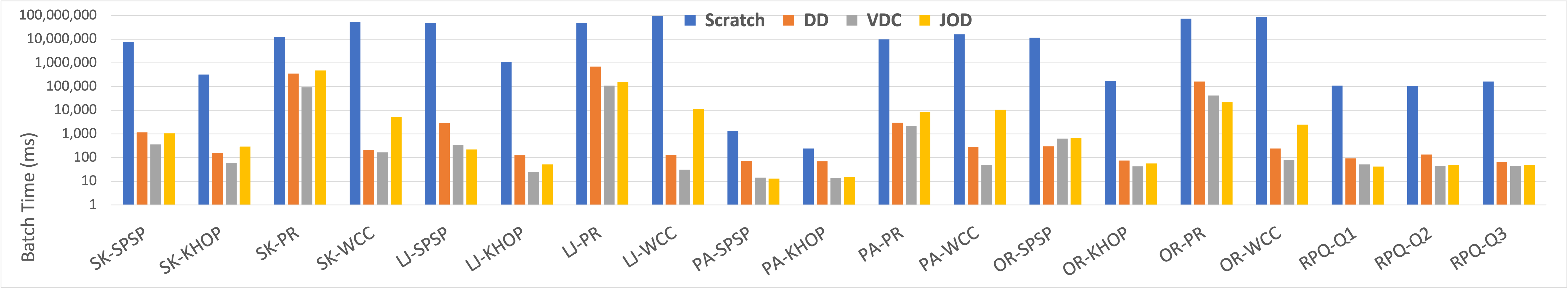}
		\caption{Total Batch Time (ms) with $50\%$ batches}
	\end{subfigure}
	\vspace{-10pt}
	\caption{Comparison between Scratch implementation (\SC), Differential Dataflow (\DD), our vanilla \DC\ implementation on top of Graphflow (\VDC), and join-on-demand (\JOD) using different ratios of edge deletions.}
	\label{fig:baseline-delete}
\end{figure*}

For further analysis, we examine our workloads with different probability of delete batches (0\%, 25\%, 50\%, 75\%, 100\%). Similar to previous experiments, we ran different queries with 100 batches, each with 1 edge. In Figure Figure~\ref{fig:delete}, we did not add ``Scratch'' because it is several order of magnitudes slower than other approaches. In general, we found that changing the ratio of batches with delete does not change our results regarding \JOD, \DET, and \PROB. An important observation, however, is that for SPSP query \VDC\ is getting slower as the delete probability increases while \JOD, \DET, and \PROB, are getting faster. This is because SPSP is a weighted query and typically has a large number of iterations and a large number of \diffs. A vanilla \DC\ implementation (\VDC) does not use early dropping, therefore deleting edges leads to adding more negative multiplicities which then adds more overhead to be stored and maintained. On the other hand, with early-dropping, deleting edges lead to reduce the number of \diffs.

\begin{figure*}
	\centering
	\begin{subfigure}{0.66\columnwidth}
		\includegraphics[width=2.2in]{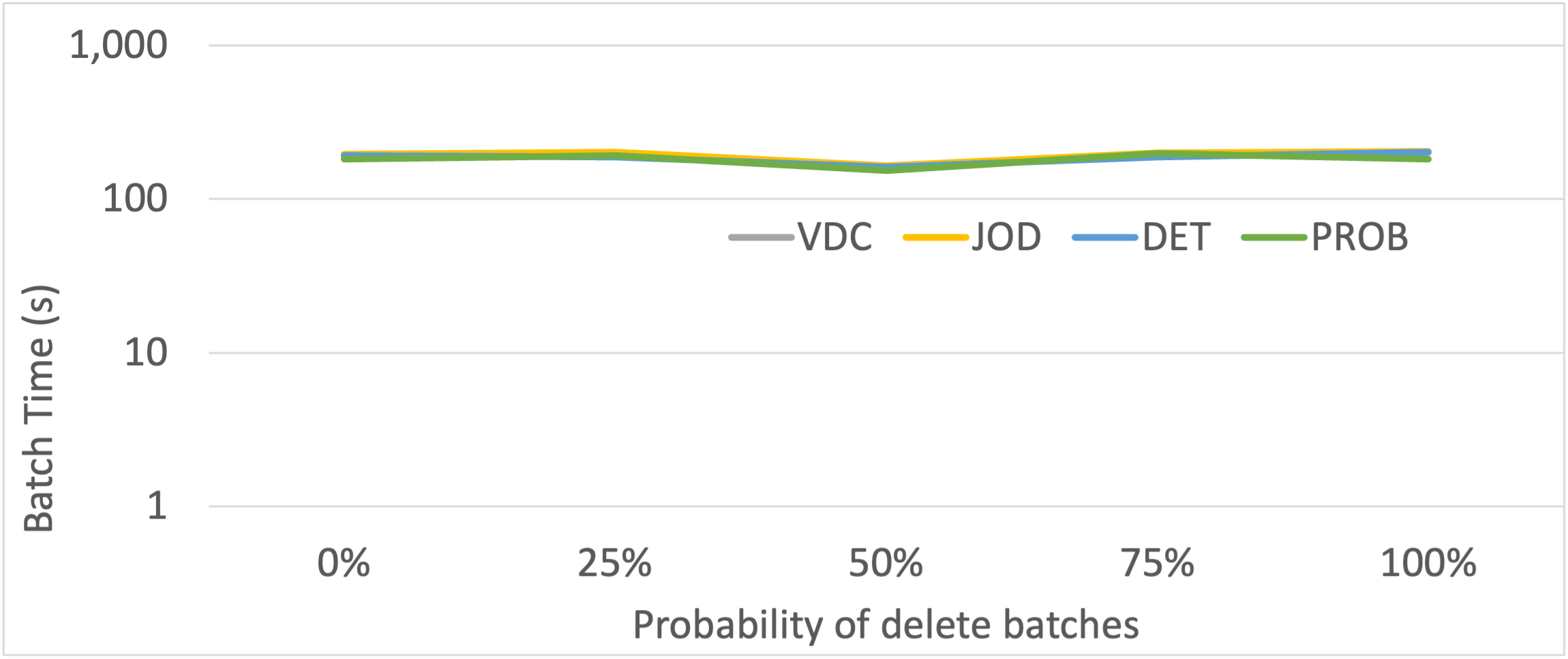}
		\caption{PagreRank}
	\end{subfigure}
	\begin{subfigure}{0.66\columnwidth}

		\includegraphics[width=2.2in]{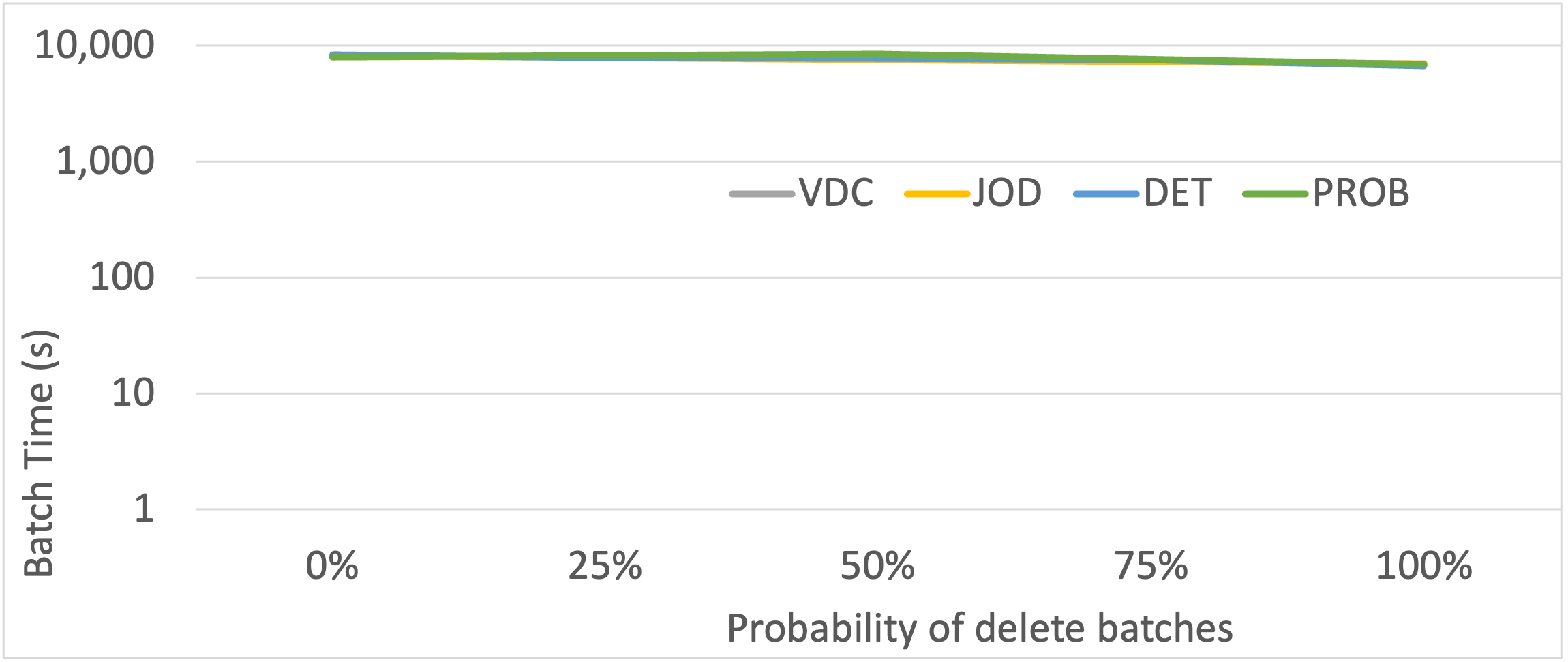}
		\caption{Connected Component}
	\end{subfigure}
	\begin{subfigure}{0.66\columnwidth}
		\includegraphics[width=2.2in]{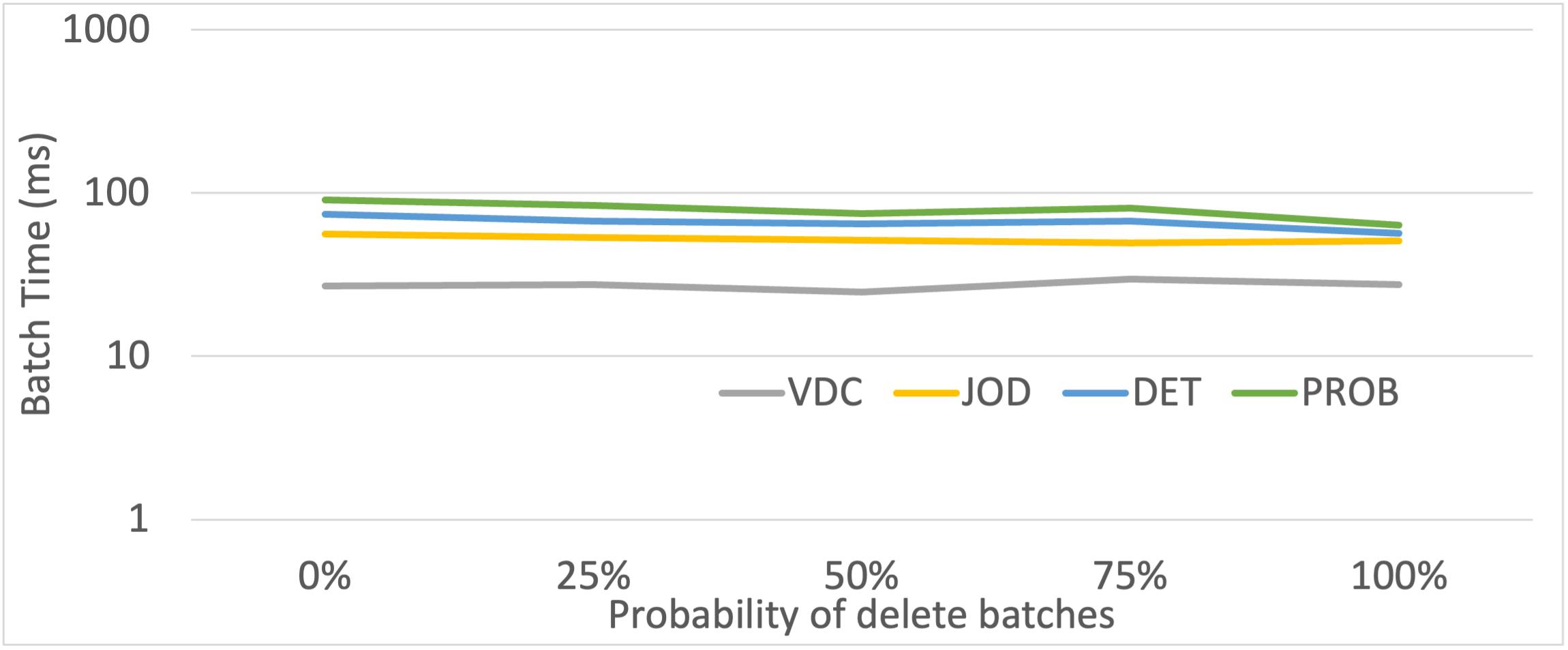}
		\caption{K-hop query}
	\end{subfigure}
	\begin{subfigure}{0.66\columnwidth}
		\includegraphics[width=2.2in]{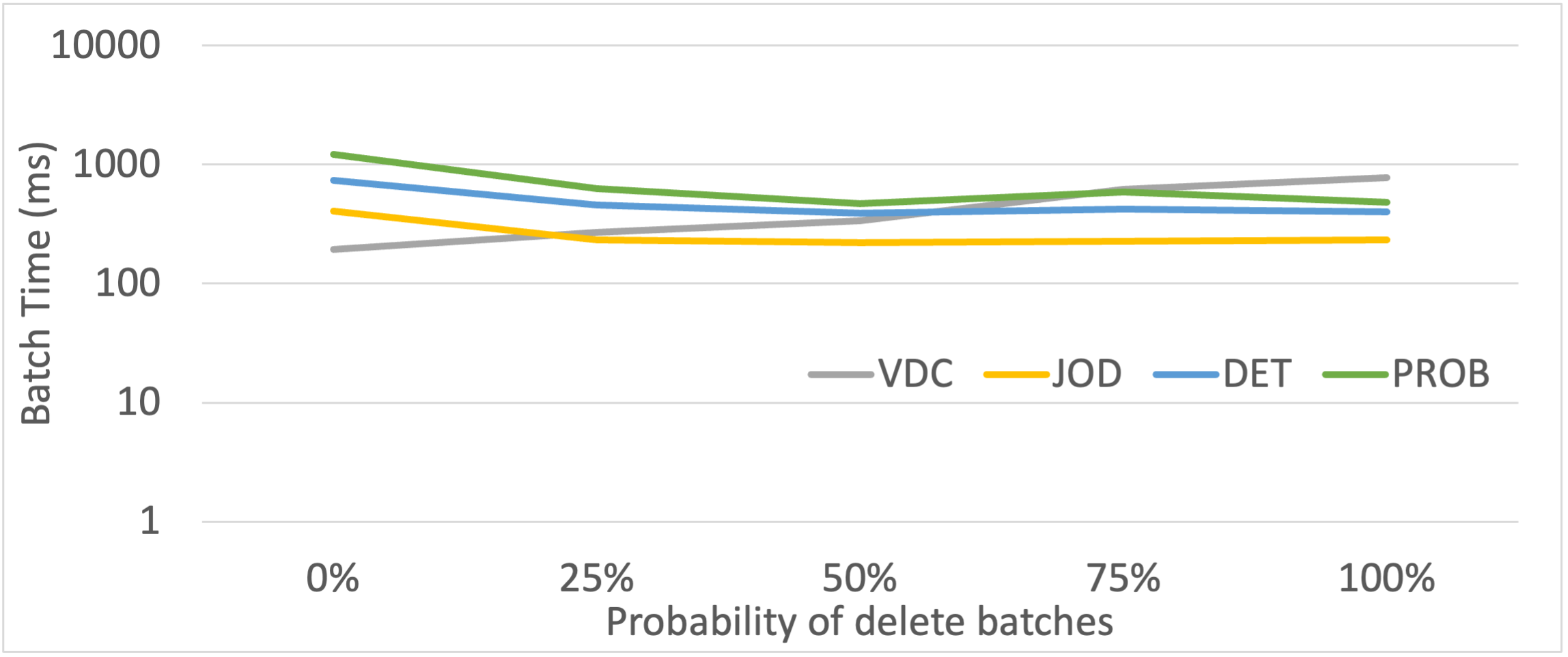}
		\caption{SPSP query}
	\end{subfigure}
	\caption{Changing the probability of delete batches while running different queries on LiveJournal dataset.}
	\label{fig:delete}
\end{figure*}

%% file: implementation-short.tex
\section{Implementation}
\label{sec:implementation}







Upon an update $\delta E_{k+1}$ to the graph, our implementation of DC and DC$^{\text{JOD}}$ keeps track of a ``frontier'',
which is the list of vertices and iteration numbers during which the aggregation operator should be re-computed. 
These are stored as an array of hash sets, to remove duplicate additions of a vertex into this set, where there is
a set for each IFE iteration until $max$. Recall that $max$ is the maximum number of iterations IFE has executed 
in any of the graph versions. Since we do eager merging, in our case $max$ is the maximum IFE iteration for $G_k$. 

We store the differences in vertex states (i.e., the output of the aggregation operator) also in a hash table where the keys
are vertex IDs and the value is a list of pairs $\langle i, s_v^{i} \rangle$ that is sorted by $i$, where $s_v^i$ is the new state of 
vertex $v$ in IFE iteration $i$. Recall again that because of eager merging our timestamps are one dimensional. To
check for the state of a vertex $v$ at iteration $i$, we find the latest available iteration $i^* \le i$ in $v$'s sorted
list using binary search.

For partial dropping approaches discussed in Section~\ref{sec:partial}, we use a separate data structure (\texttt{DroppedVT}) to store the dropped vertex-timestamp pairs. For \DET\ we use a hash table, such that the key is vertex-id and the value is a sorted list of dropped iterations. When the algorithm needs to check if a pair ($\langle v, i \rangle$) exists in \texttt{DroppedVT}, it finds the list of iterations using the vertex-id ($v$) and then search for the latest dropped iteration $d^* \le i$ in the list. For \PROB\ the hash table is replaced by a Bloom filter. Each object in this Bloom filter is $8$-bytes object and is constructed by concatenating vertex-id and iteration number together using binary operations. Searching in the Bloom filter requires constructing a search object first, using binary operations, and then check if the Bloom filter contains this object. 